\documentclass{article}
\usepackage{fullpage}

\usepackage{latexsym}
\usepackage{amsmath}
\usepackage{amssymb}
\usepackage{amsthm}
\usepackage{hyperref}
\usepackage{cite}
\usepackage{graphicx}
\usepackage{color}
\usepackage{amstext}
\usepackage{array}
\usepackage{xspace}
\usepackage{calc}

\usepackage[linesnumbered,ruled,vlined]{algorithm2e}
\usepackage[noend]{algpseudocode}

\makeatletter
\def\BState{\State\hskip-\ALG@thistlm}
\makeatother

\usepackage{tikz}
\usetikzlibrary{decorations.markings}
\tikzstyle{vertex}=[circle, draw, inner sep=0pt, minimum size=6pt]

\newtheorem{theorem}{Theorem}[section]
\newtheorem{theorem*}{Theorem}
\newtheorem{corollary}[theorem]{Corollary}
\newtheorem{lemma}[theorem]{Lemma}

\newtheorem{proposition}[theorem]{Proposition}
\newtheorem{claim}[theorem]{Claim}
\newtheorem{fact}[theorem]{Fact}

\theoremstyle{definition}
\newtheorem{definition}[theorem]{Definition}

\theoremstyle{remark}
\newtheorem{remark}[theorem]{Remark}

\numberwithin{equation}{section}

\makeatletter
\def\paragraph{\addvspace{\medskipamount}\@startsection{paragraph}{4}%
  \z@\z@{-\fontdimen2\font}%
  {\normalfont\itshape}}
\makeatother

\newcommand{\SL}{\mathrm{SL}}
\newcommand{\GL}{\mathrm{GL}}
\newcommand{\F}{\mathbb{F}}
\newcommand{\K}{\mathbb{K}}

\newcommand{\Z}{\mathbb{Z}}
\newcommand{\Q}{\mathbb{Q}}
\newcommand{\C}{\mathbb{C}}
\newcommand{\N}{\mathbb{N}}

\newcommand{\NP}{\mathrm{NP}}

\newcommand{\cB}{\mathcal{B}}
\newcommand{\cA}{\mathcal{A}}
\newcommand{\rk}{\mathrm{rk}}
\newcommand{\cork}{\mathrm{cork}}
\newcommand{\blowup}[2]{{#1}^{[#2]}}
\newcommand{\rblowup}[2]{{#1}^{\{#2\}}}

\newcommand{\trans}[1]{{#1}^{\mathrm{T}}}
\newcommand{\fdchar}{\mathrm{char}}

\newcommand{\nrk}{\mathrm{ncrk}}
\newcommand{\poly}{\mathrm{poly}}
\newcommand{\im}{\mathrm{im}}
\newcommand{\zvec}{\mathbf{0}}
\newcommand{\gcap}{\mathrm{Cap}}

\renewcommand{\L}{\mathbb{L}}
\newcommand{\algo}[1]{\textsc{#1}}

\begin{document}

\title{Non-commutative Edmonds' problem and matrix semi-invariants }

\author{
G\'abor Ivanyos
\thanks{Institute for Computer Science and Control, Hungarian 
Academy of Sciences, 
Budapest, Hungary.
{E-mail: \tt Gabor.Ivanyos@sztaki.mta.hu}.}
\and
Youming Qiao
\thanks{Centre for Quantum Computation and Intelligent Systems,
 University of Technology Sydney, Australia.
{E-mail: \tt jimmyqiao86@gmail.com.}}
\and
K. V. Subrahmanyam
\thanks{Chennai Mathematical Institute, Chennai, India.
{E-mail: \tt kv@cmi.ac.in}.}
}

\maketitle

\begin{abstract}

In 1967, J. Edmonds introduced the problem of computing the rank over the rational 
function field of an $n\times n$ 
matrix $T$ with 
integral homogeneous linear polynomials. In this 
paper, we consider the \emph{non-commutative version of Edmonds' problem}: 
compute the rank of $T$ over the free skew field. 
This problem has been proposed, sometimes in disguise, from 
several different perspectives, in the study of e.g. the free skew field itself 
(Cohn 1973), matrix spaces of low rank 
(Fortin-Reutenauer, 2004), Edmonds' original problem 
(Gurvits, 2004), and more recently, non-commutative arithmetic 
circuits with divisions (Hrube\v{s} and Wigderson, 2014). 

It is known that this problem relates to the following invariant ring, which we 
call the \emph{$\F$-algebra of matrix 
semi-invariants}, denoted as $R(n, m)$.
For a field $\F$, it is the ring of invariant polynomials for the action 
of $\SL(n, \F) \times \SL(n, \F)$ on tuples of matrices -- $(A, C)\in\SL(n, \F) 
\times \SL(n, \F)$ sends $(B_1, \dots, 
B_m)\in M(n, \F)^{\oplus m}$ to $(AB_1\trans{C}, \dots, AB_m\trans{C})$. Then 
those $T$ with non-commutative rank $<n$ correspond to those points in the 
nullcone of $R(n, m)$. In particular, if the nullcone of $R(n, m)$ is defined by 
elements of degree $\leq \sigma$, 
then there follows a $\poly(n, \sigma)$-time 
randomized algorithm to 
decide whether the non-commutative rank of $T$ is full. To our 
knowledge, previously the best bound for $\sigma$ was $O(n^2\cdot 4^{n^2})$ over 
algebraically closed fields of characteristic $0$ (Derksen, 2001). 

We now state the main contributions of this paper:
\begin{itemize}
\item We observe that by using an algorithm of Gurvits, and assuming the above 
bound $\sigma$ for $R(n, m)$ over $\Q$, deciding whether or not $T$ has 
non-commutative 
rank 
$<n$ over $\Q$ can be done \emph{deterministically} in time polynomial in the 
input size and $\sigma$.

\item When $\F$ is large enough, we devise an algorithm for the non-commutative 
Edmonds problem in time polynomial in $(n+1)!$. Furthermore, due to the structure 
of this algorithm, we also have the following results.
\begin{itemize}
\item If the commutative rank and the non-commutative rank of $T$ differ by a 
constant, then there exists a randomized efficient algorithm to compute the 
non-commutative rank of $T$. 
This improves a result of Fortin and 
Reutenauer, who gave a randomized efficient algorithm to decide whether the 
commutative and non-commutative ranks are equal. 
\item We show that $\sigma\leq (n+1)!$. This not only 
improves the bound obtained from Derksen's work 
over algebraically closed field of characteristic $0$ but, more importantly, 
also provides for the first time an explicit bound on $\sigma$ for matrix 
semi-invariants over fields of positive characteristics. Furthermore, this does 
not require $\F$ to be algebraically closed. 
\end{itemize}
\end{itemize}
\\~\\
%
\end{abstract}

\section{Introduction}

\subsection{The non-commutative Edmonds problem}

In 1967, J. Edmonds introduced the following problem~\cite{Edm67}: let $X=\{x_1, 
\dots, x_m\}$ be a set of variables. Given an 
$n\times n$ matrix $T$ 
whose entries are homogeneous linear polynomials from $\Z[X]$, determine the 
rank of $T$ over the rational function field $\Q(X)$, denoted as $\rk(T)$. The 
decision version of Edmonds' 
problem is to decide whether $T$ is of full rank or not; this decision 
version is better known now as the symbolic determinant 
identity testing (SDIT) problem. It is natural to consider this 
problem over any field $\F$. If $|\F|$ is constant, this problem is $\NP$-hard 
\cite{BFS}. This is not the setting we are concerned with -- we will always 
assume 
$|\F|$ to be at least $\Omega(n)$.

When $|\F|\geq 2n$, the Schwartz-Zippel lemma provides a randomized efficient 
algorithm. To devise a deterministic efficient algorithm has a long history, and 
is of fundamental importance in complexity theory. Originally, the main motivation 
was its applications to certain combinatorial problems, most notably the maximum 
matching problem on graphs, 
as exploited by Tutte \cite{Tutte}, Edmonds \cite{Edm67}, Lov\'asz \cite{Lovasz}, 
among others.\footnote{In these applications, $T$ is usually of certain specific 
forms, for example, as a mixed matrix: each entry is a single variable or a field 
element, and each variable appears only once. }
Since 2003, a major incentive to study SDIT arises from its implications 
to circuit lower bounds, as shown in the wonderful work by Kabanets and 
Impagliazzo \cite{KI04}. Improving the results in \cite{KI04}, Carmosino et al. 
showed that such an algorithm implies the existence of an explicit multilinear 
polynomial family such that its graph
is computable in $\mathrm{NE}$, but the polynomial family cannot be computed by 
polynomial-size arithmetic circuits \cite{CIKK15}.

In this paper, we study Edmonds' problem in the non-commutative 
setting. In other words, we view the entries of $T$ as elements of $\F\langle 
X\rangle$, 
the algebra of non-commutative polynomials over $\F$. To state this, we need a 
non-commutative counterpart of the rational function field. Note that, due to 
non-commutativity, the best we can 
hope for is a skew field (a.k.a.~a non-commutative field or a division ring). The 
\emph{free skew field} is the non-commutative analogue of the rational function 
field. For $T$ a matrix with homogeneous 
linear polynomials, Fortin and Reutenauer \cite{FR04} defined the 
\emph{non-commutative rank} of $T$, denoted as $\nrk(T)$, as its rank over the 
free skew field. By the \emph{non-commutative Edmonds problem} we 
mean the problem of computing $\nrk(M)$, and by the \emph{non-commutative 
full rank problem} (NCFullRank) we mean the problem of deciding whether 
$\nrk(M)$ is full or not. 

Since we will not define the free skew field formally, we provide 
a definition of $\nrk(T)$ that is due to Cohn \cite{Cohn95}: $\nrk(T)$ is the 
minimum $s\in\Z^+$ s.t. $T$ can be written as $PQ$, 
where $P$ and $Q$ are matrices with homogeneous linear polynomials from 
$\F\langle 
X\rangle$, and of size $n\times s$ and $s\times n$, respectively. It may also be 
worthwhile to recall that, since every module over a skew field is free, the 
row (resp. column) rank of a matrix over a skew field can be defined as the 
rank of the module generated by the rows (resp. columns). Then just as the 
case of matrices over fields, it can be shown that row rank and column rank are 
equal, so either of them defines the rank of a matrix over a skew field.

The free skew field was first constructed by Amitsur \cite{Ami66}, and alternative 
constructions were subsequently 
given by Bergman \cite{Berg70}, Cohn \cite{Cohn}, and Malcolmson \cite{Mal78}. 
We refer the reader to \cite{HW15} by Hrube\v{s} and Wigderson for a nice 
introduction to the free skew field from the perspective of algebraic 
computations. 
Cohn's books \cite{Cohn,Cohn95} serve as a comprehensive 
introduction to this topic. 

It will be clear soon that $\rk(T)\leq \nrk(T)$, and Fortin and 
Reutenauer showed that $\nrk(T)\leq 2\rk(T)$, and exhibited an example $T$ for 
which $\nrk(T)=3/2\cdot \rk(T)$ \cite{FR04}. 
In \cite{CR99}, Cohn and Reutenauer presented an algorithm to decide whether 
$\nrk(T)$ is full or not\footnote{As remarked in \cite{FR04}, this algorithm can 
be generalized to compute $\nrk(T)$. }, which puts this problem in 
$\mathrm{PSPACE}$ since it reduces to testing the solvability of a system of 
multivariate polynomial equations. 
Unlike its commutative counterpart, it is not even
clear that the non-commutative Edmonds problem has a randomized efficient 
algorithm. In Section~\ref{subsec:inv}, we will discuss a natural randomized 
algorithm for NCFullRank, but its efficiency will depend on an
invariant-theoretic quantity. 

\subsection{Equivalent formulations of the non-commutative Edmonds 
problem}\label{subsec:eq}

Like Edmonds' problem, its non-commutative counterpart also has a long 
history, though in the literature, it often stated in a very different way. In 
1973, Cohn first studied this problem from the perspective of understanding the 
free skew field, and showed it to be decidable in \cite{Cohn73,Cohn75}.
In 2003, Gurvits posed this problem in his remarkable work on Edmonds' problem 
\cite{Gurvits}. Recently, Hrube\v{s} and Wigderson arrived at this problem in 
their study of non-commutative arithmetic circuits with divisions \cite{HW15}. 
Indeed, a very intriguing 
feature of the non-commutative Edmonds problem is the existence of several 
interesting equivalent formulations. Instead of relying on the free skew field, 
these formulations use 
either linear algebra, or concepts from invariant theory, or quantum 
information theory. They are scattered in the literature, so we collect them here,
to illustrate the various facets of this problem, introduce some previous works, 
and motivate the study of the non-commutative Edmonds problem. 

To state these formulations we need some notations. $M(n, \F)$ denotes the 
linear space of 
$n\times n$ matrices over $\F$. A linear subspace of $M(n, \F)$ is called a 
\emph{matrix space}. Given $T$, a matrix of linear forms in variables 
$X=\{x_1, \dots, x_m\}$, write $T=x_1B_1+x_2B_2+\dots+x_mB_m$ where 
$B_i\in M(n, \F)$. 
Let $\cB:=\langle B_1, \dots, B_m\rangle$, where $\langle \cdot \rangle$ denotes 
linear span. The rank of $\cB$, denoted as $\rk(\cB)$, is defined as 
$\max\{\rk(B)\mid B\in\cB\}$. We call $\cB$ \emph{singular}, if $\rk(\cB)<n$. When 
$|\F| > n$, as we will assume throughout, 
$\rk(T)=\rk(\cB)$.
\footnote{As when the field size is large 
enough, the complement of the zero set of a nonzero polynomial is non-empty. } We 
shall soon see
that $\nrk(T)$ corresponds to some property 
of $\cB$ as well, so that we can translate the study of commutative and 
non-commutative ranks of $T$ entirely to the study of $\cB$. 

Some of these formulations make sense only subject to certain 
conditions. In such cases we indicate the conditions needed before that 
formulation. 

\begin{enumerate}
\item Given $\cB=\langle B_1, \dots, B_m\rangle \leq M(n, \F)$, a subspace $U\leq 
\F^n$ is called a 
\emph{$c$-shrunk subspace} of $\cB$, for $c\in\N$, if there exists $W\leq \F^n$, 
such that $\dim(W)\leq \dim(U)-c$, and for every $B\in\cB$, $B(U)\leq W$. $U$ is 
called 
a shrunk subspace of $\cB$, if it is a $c$-shrunk subspace for some $c\in\Z^+$.

Question: compute the maximum $c$ such that there exists a $c$-shrunk subspace. 

Remark: Cohn showed that the non-commutative rank is not full if and only if there 
is a shrunk subspace \cite{Cohn95}. This was generalized by Fortin and 
Reutenauer\cite[Theorem 1]{FR04} who showed a precise relationship  between 
non-commutative rank and the existence of $c$-shrunk subspaces. Their motivation 
to consider this problem was 
to connect matrices over linear forms on the one hand, and  
matrix spaces of low rank on the other. The latter topic was studied in e.g.  
\cite{PrimitiveI,EH88}. 
By \cite{FR04}, we can define the \emph{non-commutative rank} of 
$\cB$ as 
$$n-\max\{c\in\{0, 1, \dots, n\} \mid \exists c\text{-shrunk subspace of }\cB\},$$
and it follows that $\nrk(\cB)=\nrk(T)$. So we may (and do) identify $T$ with 
$\cB$ in the following. 
\item ($\F$ is large enough) Given $\cB=\langle B_1, \dots, B_m\rangle \leq
M(n, \F)$, the $d$th tensor 
blow-up of $\cB$, is $\blowup{\cB}{d}:=M(d, \F)\otimes \cB\leq M(dn, \F)$. It is 
clear that 
$\rk(\blowup{\cB}{d})\geq d\cdot \rk(\cB)$. We shall prove 
that when $\F$ is large enough, then $d$ always divides $\rk(\blowup{\cB}{d})$. 
Furthermore, when $d>n$, then 
$\rk(\blowup{\cB}{d+1})/(d+1)\geq \rk(\blowup{\cB}{d})/d$. 
See Lemma~\ref{lem_reg_blowup}, Corollary~\ref{cor:regularity}, and 
Remark~\ref{remark:coprime}.

Question: compute $\lim_{d\to\infty}\rk(\blowup{\cB}{d})/d$.

Remark: That NCFullRank is equivalent to deciding whether 
$\rk(\blowup{\cB}{d})=nd$ for some $d$ was shown by Hrube\v{s} and 
Wigderson \cite{HW15}. Our formulation here is a straightforward quantitative 
generalization of 
their statement. Hrube\v{s} and Wigderson's motivation was to study 
non-commutative arithmetic formulas \emph{with divisions}. 

\item ($\F=\C$) 
Given $B_1, \dots, B_m \in M(n, \Q)$, construct a completely 
positive operator $P: M(n, \C)\to M(n, \C)$, sending $A\to 
\sum_{i\in[m]}B_iAB_i^{\dagger}$. For $c\in\N$, $P$ is called rank 
$c$-decreasing, if there exists a positive semidefinite $A$, such that 
$\rk(A)-\rk(P(A))=c$. 

Question: compute the maximum $c$ such that $P$ is rank $c$-decreasing. 

Remark: Gurvits stated the problem of deciding whether $P$ is rank 
non-decreasing or not~\cite{Gurvits}. His original motivation was to study 
Edmonds' original (commutative) problem, and the main result in \cite{Gurvits} 
solves the case when the commutative and non-commutative rank coincide (see 
Theorem~\ref{thm:gurvits}). 

\item (NCFullRank) Consider the action of $(A, 
C)\in\SL(n, \F)\times \SL(n, \F)$ on a tuple of matrices $(B_1, \dots, B_m)\in 
M(n, \F)^{\oplus m}$ by 
sending it to $(AB_1\trans{C},$ $\dots, AB_m\trans{C})$.\footnote{This action can 
also be written as: $(A, C)$ sending $(B_1, \dots, B_m)$ to $(AB_1C^{-1}, \dots, 
AB_mC^{-1})$. We adopt the transpose rather than the inverse, as the transpose 
yields a polynomial representation rather than a rational representation. 
Furthermore when Derksen's result is applied to the transpose, it gives a somewhat 
better bound (Fact~\ref{thm:derksen}).} Let $R(n, m)$ be the $\F$-algebra 
of invariant 
polynomials with respect to this action. 
The \emph{nullcone} of $R(n, m)$ is the common zero of all homogeneous 
positive-degree 
polynomials in $R(n, m)$. 

Question: decide whether or not $(B_1, \dots, B_m)$ is in the nullcone of 
$R(n,m)$. 

\end{enumerate}
That the original formulation is equivalent to (1) comes
from~\cite{FR04}. The equivalence between (1) and (3) is straightforward. The 
equivalence among decision 
versions of (1) and (2), and (4) can be obtained via the ring of matrix 
semi-invariants, as described in Section~\ref{subsec:inv}. One way to prove the 
equivalence between (1) and (2) is via Theorem~\ref{thm-blup-main}.

To summarize, the non-commutative Edmonds problem can be derived 
naturally from the perspectives of quantum information theory,\footnote{It remains 
to investigate the physical meaning for a super-operator to be rank non-decreasing 
though. } and invariant theory. It is of great interest in non-commutative 
algebraic computation with divisions, and in the study of matrix spaces of low 
rank. Our
motivation to study this is because a solution to the non-commutative Edmonds 
problem will throw light on 
its commutative counterpart. 
Shrunk subspaces form a natural and important witness for the 
singularity of a matrix space. 
Therefore, if the non-commutative Edmonds problem can be solved 
deterministically in polynomial time, it means that, 
for SDIT, the bottleneck lies in recognizing those 
singular matrix spaces
without such witnesses. This connection will be detailed in 
Section~\ref{sec:gurvits}. 

\subsection{Matrix semi-invariants}\label{subsec:inv}

Formulation (1), (2) and (4) have a common origin, namely the invariant ring $R(n, 
m)$ described in (4). We shall call $R(n, m)$ \emph{the ring of matrix 
semi-invariants}, as (1) 
it is closely related to the classical ring of matrix invariants \cite{Pro76} 
(see below for 
the definition, and \cite{Domokos00,ANK07} for the 
precise relationship between these two rings); and (2) it is the ring of 
semi-invariants of the representation of the 
$m$-Kronecker quiver with dimension vector $(n, n)$. Here, the $m$-Kronecker 
quiver is the quiver with two vertices $s$ and $t$, 
and $m$ arrows pointing from $s$ to $t$. When $m=2$, it is the classical 
Kronecker quiver. 
The reader is referred to \cite{DW00,SV01,DZ01} for a 
description of the semi-invariants for arbitrary quivers. 

%

The equivalence between (1) and (4) comes from the observation that the
$(B_1, \dots,$ $B_m)$ with a shrunk subspace are exactly the points in the 
nullcone 
of $R(n, m)$ \cite{BD06,ANK07}. The equivalence between (2) and (4) can be seen
from the first fundamental theorem (FFT) of matrix semi-invariants 
\cite{DW00,SV01,DZ01,ANK07}. To describe this we need some 
notations: for $n\in\N$, $[n]:=\{1, \dots, n\}$. Note that $R(n, 
m)\subseteq \F[x_{i,j}^{(k)}]$ where $i, j\in[n]$, $k\in [m]$, and $x_{i,j}^{(k)}$ 
are independent variables. Let $X_k=(x_{i,j}^{(k)})_{i,j\in[n]}$ be a matrix of 
variables. Then for $A_1, 
\dots, A_m \in M(d, \F)$, $\det(A_1\otimes X_1+\dots+A_m\otimes X_m)$ is a matrix 
semi-invariant, and every matrix semi-invariant is a linear combination of such 
polynomials. Therefore, $(B_1, \dots, B_m)$ is in the nullcone, if and only if for 
all $d\in\Z^+$ and all $(A_1, \dots, A_m)\in M(d, \F)^{\oplus m}$, $A_1\otimes B_1
+\dots+A_m\otimes B_m$ is singular. 


It is well-known that the matrix semi-invariant ring is finitely generated, by 
Hilbert's celebrated work \cite{Hilbert}. This implies that there exists some  
integer $d$ such that those matrix semi-invariants of degree no more than $d$ 
define $R(n, m)$. This motivates the following definition. 
\begin{definition}
$\beta(R(n,m))$ is the smallest integer $d$ such that 
$R(n,m)$ is generated by invariants of degree  $\leq d$.
\end{definition} 

An explicit upper bound on 
$\beta(R(n, m))$ turns out to be 
particularly interesting for 
the purpose of the NCFullRank problem. As already suggested by 
Hrube\v{s} and Wigderson \cite{HW15}, if 
$R(n, m)$ has a degree bound $\beta=\beta(R(n, m))$, one can do the 
following: take $m$ $d \times d $ variable matrices $Y_1, 
\dots, Y_m$, $Y_k=(y^{(k)}_{i,j})$ and form the polynomial 
\begin{equation}\label{eq:blup_poly}
\det(Y_1\otimes B_1+\dots+Y_m\otimes B_m) \in \F[y^{(k)}_{i,j}]_{k\in[m], i, 
j\in[d]}.
\end{equation}
Letting $d$ go from $1$ to $\beta$, this system of polynomials 
characterizes $\nrk(T) < n$: $\nrk(T)<n$ if and only if all these polynomials 
are the zero polynomial. This immediately gives a randomized algorithm for 
NCFullRank over large enough fields, with time complexity $\poly(n, \beta)$. 

In fact, for the above application, what really matters is another important bound 
$\sigma=\sigma(R(n, m))$. This is defined as the minimum integer $d$ with the 
property 
that $(B_1,\ldots,B_m)\in M(n, \F)^{\oplus m}$ is in the nullcone if and only if  
all polynomials of degree 
$\leq d$ in $R(n, m)$ vanish on $\{B_1,\ldots,B_m\}$. It is clear that $\sigma\leq 
\beta$, and the
above reasoning  goes through when $\beta$ is replaced by $\sigma$.

Over algebraically closed fields of characteristic $0$, by directly 
employing Derksen's bounds for invariant rings satisfying 
certain general conditions \cite{derksen_bound}, the following bound can be 
derived. For completeness we include a proof in Appendix~\ref{app:derksen}.
\begin{fact}[{\cite{derksen_bound}}]\label{thm:derksen}
Over algebraically closed fields of characteristic $0$, for 
$R(n, m)$, 
$\beta\leq\max\{2, 3/8\cdot n^4\cdot \sigma^2\}$, and $\sigma\leq 1/4\cdot 
n^2\cdot 
4^{n^2}$.
\end{fact}

In particular,  if  $\sigma$ is polynomial in $n$, 
then $\beta$ is polynomial in $n$ as well. 

It is generally believed that over fields of characteristic $0$, 
the bounds we get for $R(n, m)$ using Derksen's theorem is far from optimal. One 
reason to believe so is that 
$R(n,m)$ is closely related
to another ring of invariants: let 
$A\in\SL(n, \F)$ act on $(B_1, \dots, B_m)\in M(n, \F)^{\oplus m}$ by 
simultaneous conjugation -- i.e. $A$ sends the tuple to $(AB_1A^{-1}, \dots, 
AB_mA^{-1})$. 
Denoted by $S(n, m)$, this is just the classical ring of matrix
invariants \cite{Pro76}. The structure of $S(n, m)$ is 
well-understood. Over fields of characteristic $0$, the first and second 
fundamental theorems 
for $S(n,m)$, and an 
$n^2$ 
upper bound for $\beta(S(n, m))$ were established in 1970's, by the works of 
Procesi, Razmysolov, and Formanek  
\cite{Pro76,Raz74,formanek_gen}. See \cite{Domokos00,ANK07} for the 
precise relationship between the rings $R(n,m)$ and 
$S(n,m)$. Note that when applied to $S(n, m)$ over 
characteristic $0$, Derksen's bound yields $\beta(S(n, m))\leq \max\{2, 3/8\cdot 
n^2\cdot \sigma^2\}$ and $\sigma(S(n, m))=n^{O(n^2)}$, far from the $n^2$ 
bound given above. 

Another reason to believe that Derksen's bounds are far from optimal is that for 
certain small $m$ or $n$, explicit 
generating sets of $R(n, m)$ 
have been computed in e.g. \cite{Domokos00,Domokos00_2,DD12,IQS2}. In these cases, 
elements of degree $\leq n^2$ generate the 
ring.\footnote{We thank M. Domokos for pointing out this fact to us.} 

If we turn to positive characteristic fields then, to our best knowledge, no 
explicit bounds for $\beta(R(n, m))$ nor $\sigma(R(n, m))$ have been derived. Note 
here that the relation 
between $\beta$ and $\sigma$ as in Fact~\ref{thm:derksen} is not known to hold, 
due to the assumption on the field properties there. This case is 
important, for example, in the application to identity 
testing, and 
division 
elimination for non-commutative arithmetic formulas with divisions over fields of 
positive characteristics \cite{HW15}. For $S(n,m)$ over fields of 
positive 
characteristics, the FFT was established by Donkin in \cite{Donkin92,Donkin93}. 
Over fields of positive characteristic an $O(n^3)$ upper bound for 
$\sigma(S(n,m))$ can be 
derived from \cite[Proposition 9]{CIW}, and in 
\cite{Domokos02a,Domokos02b} Domokos proved an upper bound $O(n^7 m^n)$ on 
$\beta(S(n,m))$.

\subsection{Our results}

In the previous sections, we defined the non-commutative Edmonds problem and 
the NCFullRank problem, and illustrated their connections to matrix 
semi-invariants. Indeed, our results suggest that progress on one topic 
helps to advance the other as well. 

The first result shows that an upper bound for $\sigma(R(n, m))$ actually implies 
a \emph{deterministic} algorithm for NCFullRank over $\Q$, rather than 
just a randomized one as in Section~\ref{subsec:inv}. 

\begin{proposition}\label{thm:sdit_main}
Over $\Q$, if the nullcone of $R(n, m)$ is defined by elements of degree $\leq 
\sigma=\sigma(n, 
m)$, then there exists a deterministic algorithm that solves NCFullRank with bit 
complexity polynomial in $\sigma$ and the input size. 
\end{proposition}

In particular, if $\sigma$ is a polynomial in $n$ and $m$, then NCFullRank can be 
solved deterministically in polynomial time over $\Q$. 
The key ingredient here is Gurvits' algorithm for the Edmonds' problem, although 
that algorithm works only under a promise\cite{Gurvits}. The 
distinction between deterministic and probabilistic is important: as 
illustrated at the end of Section~\ref{subsec:eq}, our original motivation of 
studying NCFullRank is to 
gain an understanding of SDIT, for which the question is to devise 
deterministic efficient algorithms. 

Our main result is an algorithm that solves the non-commutative Edmonds
problem using formulation (2). To ease the presentation, we give an informal 
statement of the main theorem in Section~\ref{sec:blowup} 
(Theorem~\ref{thm-blup-main}) here, and discuss its two consequences. 
\begin{theorem}[Theorem~\ref{thm-blup-main}, informal]\label{thm:main_informal}
Given a matrix space $\cB\leq M(n, \F)$ over a large enough field, there exists a 
deterministic algorithm that computes $\rk(\cB)$ using $\poly((n+1)!)$ many 
arithmetic operations. 
Over $\Q$ the algorithm runs in time polynomial in the bit size of the input and 
$(n+1)!$.\footnote{All algorithms presented in this paper, when 
working over $\Q$, have bit complexity polynomial in the input size, and some 
additional parameters. Sometimes, we may omit the input size but 
focus on those more important parameters.} 
\end{theorem}

In \cite{FR04}, Fortin and Reutenauer asked 
for ``an algorithm which uses only linear-algebraic techniques.'' Indeed, the 
algorithm for Theorem~\ref{thm:main_informal} may be viewed as one, though it 
relies on certain routines dealing with objects from cyclic field extensions and 
division algebras. 

Two interesting consequences now 
follow. Firstly, we have a randomized efficient algorithm to compute the 
non-commutative rank if it differs from the commutative rank by a constant. 
(Recall 
that $\rk(\cB)\leq \nrk(\cB)\leq 2\rk(\cB)$.) 
Its easy proof is put after the 
statement of Theorem~\ref{thm-blup-main}.
\begin{corollary}\label{cor:blup_edmonds}
For $\cB\leq M(n, \F)$, let $c=\nrk(\cB)-\rk(\cB)$, and assume $\F$ is 
of size $\Omega(n\cdot (n+1)!)$. Then the non-commutative rank of $\cB$ can be 
computed probabilistically in time polynomial in $(n+1)^{c+1}$.
\end{corollary}

Secondly, we immediately obtain an explicit bound 
for $\sigma(R(n, m))$ as a consequence of Theorem~\ref{thm-blup-main}. 
By Fact~\ref{thm:derksen}, we also get a bound 
on $\beta(R(n, m))$, over an algebraically closed field of characteristic $0$. Its 
proof is also put after Theorem~\ref{thm-blup-main}.

\begin{corollary}\label{cor:blup_degree}
Over any field $\F$ of size $\Omega(n\cdot (n+1)!)$, $\sigma(R(n, m))\leq 
(n+1)!$. If furthermore $\F$ is of characteristic $0$ and  
algebraically closed, then $\beta(R(n, m))\leq \max\{2, 3/8\cdot n^4\cdot 
((n+1)!)^2\}$. 
\end{corollary}

This improves the bounds in 
Fact~\ref{thm:derksen} over algebraically closed fields of characteristic $0$. 
More importantly, to the best of our knowledge, this provides an explicit 
bound for $\sigma(R(n, m))$ over fields of positive 
characteristic for the first time. Furthermore to get this bound we only assume 
the field size to be large enough, whereas 
Fact~\ref{thm:derksen} requires our field to be algebraically closed. 

While the improvement from $2^{O(n^2)}$ to 
$2^{O(n\log n)}$ is modest, we believe it is nonetheless an interesting 
improvement from the technical point of view: note that the dimension of $\SL(n, 
\F)$ is $n^2-1$. In the line of research for bounds of an invariant ring $R$
with respect to a group $G$ (cf. \cite{popov,derksen_bound}), the dimension of $G$ 
has to 
stand on the exponent for $\sigma(R)$, and to get a bound as $2^{o(\dim(G))}$ 
seems difficult there. Furthermore, the
idea of using correctness of algorithms to get bounds on quantities of interest in 
invariant theory seems new, and may deserve to be explored further.

We also obtain certain structural results for $R(n, m)$, which are 
reported in \cite{IQS2}.

\subsection{More previous works}\label{subsec:previous}

\paragraph{Connections between invariant theory and complexity theory.} The 
results in this paper suggest a new link between invariant theory and 
complexity theory. Connections between the two fields have been emerging in recent 
years. We have already alluded to the direct connection with non-commutative 
arithmetic 
circuits, in the work of Hrube\v{s} and Wigderson \cite{HW15} above. 
In 
a series of papers titled geometric complexity theory 
(GCT) \cite{GCT1,GCT2} (see also \cite{GCT_JACM,GCT_overview}), Mulmuley and 
Sohoni pointed out possible deep 
connections between problems in invariant theory and complexity theory. GCT 
addresses the 
fundamental lower bound problems in complexity theory, e.g. the permanent versus 
determinant problem, by linking them to problems in representation theory and 
algebraic geometry. In particular, 
in \cite{GCT5}, Mulmuley established a tight 
connection between derandomizing the Noether normalization lemma, and black-box 
derandomizing the polynomial identity test. The degree bounds of various invariant 
rings are of central importance in that work. We briefly remark that a polynomial  
bound for $\beta(R(n, m))$, if proven, will yield similar results as what the 
$n^2$ degree bound for $S(n, m)$ has yielded in \cite{GCT5}.

\paragraph{More previous works on Edmonds' problem.} Some earlier work on this 
problem was cited at the beginning of this article. Here we mention more related 
work.

An interesting instance of Edmonds' problem is the module isomorphism problem.
Specifically, assume that we are given two
$n$-dimensional modules $U$ and $U'$ for the free algebra $\mathcal{A}$ 
over $\F$ with
$k$ generators as $k$-tuples $G_1,\ldots,G_k$ and $G_1',\ldots,G_k'$
of $n$ by $n$ matrices. Then $\mbox{Hom}_{\mathcal{A}}(U,U')$ is the $\F$-linear
subspace of $\mbox{Hom}_\F(U,U')$, identified with $M(n,\F)$, consisting
of matrices $X$ with $XG_i=G_i'X$ ($i=1,\ldots,k$). As these conditions
are linear in the entries of $X$, the space $\mbox{Hom}_{\mathcal{A}}(U,U')$
can be obtained by solving a system of homogeneous linear equations
in $n^2$ elements. Furthermore, $U'$ is isomorphic to $U$ if
and only if there exists a nonsingular matrix in $\mbox{Hom}_{\mathcal{A}}(U,U')$.
In turn, any such nonsingular matrix witnesses an isomorphism and,
by the Schwartz-Zippel lemma, for sufficiently large base field
a random homomorphism will be an isomorphism. Due to the special
algebraic structure behind this problem, it can be solved even
by deterministic polynomial-time methods, see the method of
Chistov, Ivanyos and Karpinski
\cite{CIK97} working over many fields, or a different 
approach of Brooksbank and Luks \cite{BL08} which works
over arbitrary fields, and an extension of the first method
to arbitrary fields given by Ivanyos, Karpinski and Saxena
in \cite{IKS}. Interestingly, the general case of finding
a surjective or injective homomorphism  
between non-isomorphic
modules deterministically
turns out to be as hard as the constructive version
of Edmonds' general problem~\cite{IKS}.

Recall that one motivation to study Edmonds' problem is due to its implications to 
certain combinatorial problems. This line of research mostly focuses in the case 
when the given 
matrices are of particular form, e.g. rank-1 and certain generalizations 
\cite{Geelen,Murota,HKM05,IKS} as used in bipartite graph matchings, or skew 
symmetric rank-2 and certain generalizations \cite{Geelen00,GIM03,GI05} as used in 
general graph matchings. 

Another line of research deals with matrix spaces that satisfy certain properties. 
Note that properties of matrix spaces should not depend on a 
particular basis. For example, we can define a property of matrix 
spaces as, ``having a basis consisting of rank-$1$ matrices.'' So if $\cB$ has a 
basis consisting of rank-$1$ matrices, $\cB$ may not necessarily be presented 
using this rank-$1$ basis. We are not aware of any result on the complexity 
of finding rank-$1$ generators for rank-$1$ spanned matrix spaces, if it is given 
by a basis consisting of not necessarily rank-$1$ matrices.
We believe that the problem is hard. Thus the results in 
\cite{Geelen,Murota,HKM05,IKS}, which assume that the
input is given 
by a rank-$1$ basis, do not translate to algorithms 
for rank-$1$ spanned matrix spaces. 

As far as we are aware, there are two references for SDIT 
which assume only properties of matrix 
spaces. The first one is Gurvits' algorithm in \cite{Gurvits}; this algorithm 
works over $\Q$, 
and assumes a property which Gurvits called 
``Edmonds-Rado.'' His 
algorithm, put in the context of this paper, is rephrased as 
Theorem~\ref{thm:gurvits}. Gurvits left open the problem of developing a 
deterministic efficient algorithm for rank-$1$ spanned matrix spaces over finite 
fields. This was settled in affirmative in \cite{conf_version}, the other 
reference 
that assumes properties of matrix spaces. 

Recall that the other major incentive to study Edmonds' problem is to understand 
arithmetic circuit lower bounds via \cite{KI04,CIKK15}. We believe that for this 
goal, a 
better 
indication of 
progress is to use properties of matrix spaces, rather than properties of 
the given matrices. One reason is that, whether a matrix space contains a 
nonsingular matrix, is a property of matrix spaces. Another reason is that many 
properties of matrix spaces seem difficult to test algorithmically.
Furthermore, note that in this paper we heavily rely on algorithmic techniques 
developed in \cite{Gurvits} and \cite{conf_version}. This may be viewed as another 
evidence of the importance of working with properties of matrix spaces.

\paragraph{Connections to Kronecker coefficients.} 

Recently, there was an interest in studying the semi-invariants of the 
$m$-Kronecker quivers due to its connection with the Kronecker coefficients 
\cite{ANK07,ANK08,Man10}, namely the multiplicities in the direct 
sum decompositions of the tensor products of two 
irreducible representations of symmetric groups. Giving a positive 
combinatorial description of these coefficients is considered to be one of the 
most important problems in the combinatorial representation theory of symmetric 
groups.

\subsection{Update on recent progress} 
There have been some exciting developments since we posted
a version of this paper on the arXiv.

First, Garg et 
al. presented a deterministic polynomial-time algorithm for computing the 
non-commutative rank over $\Q$ \cite{GGOW}. This is achieved via a closer analysis 
of Gurvits' 
algorithm \cite{Gurvits}. Their analysis uses the exponential bounds on 
$\sigma(R(n, m))$ as proved in this paper, or deducible from Derksen's result 
\cite{derksen_bound}. It should be noted that that the given algorithm is not 
constructive, 
in that it fails to produce a witness (e.g. shrunk subspaces). 

Second, Derksen and Makam proved
that $\sigma \leq n^2 -n$ over large enough fields~\cite{DM2}. 
Over fields of characteristic zero this implies $\beta(R(n,m))=O(n^6)$, 
settling the question of whether there is a polynomial degree bound on the 
generators for this ring of invariants.  To prove 
the upper bound on $\sigma(R(n,m))$, Derksen and Makam discover a concavity 
property of blow-ups, and rely crucially on Lemma~\ref{lem:reg_non_cons} 
proved in this paper.

After \cite{DM2} appeared, in \cite{IQSnote} we show that the technique of Derksen 
and Makam can be constructivized. By combining that with the techniques in this 
paper, we obtain a 
constructive deterministic polynomial-time algorithm for computing the 
non-commutative rank over large enough 
fields. 
There we also present another independent proof of 
$\sigma(R(n, m))\leq 
n^2+n$. That argument also builds on Lemma~\ref{lem:reg_non_cons} from this paper 
and is much 
simpler than the concavity argument of Derksen and Makam. 

\paragraph{Organization.} In Section~\ref{sec:prel} we present certain 
preliminaries. In Section~\ref{sec:gurvits} we give an exposition of the 
natural connection 
between commutative and the non-commutative Edmonds problem, and prove 
Proposition~\ref{thm:sdit_main}. In Section~\ref{sec:cda} we present an efficient 
construction of division algebras, to be used in proving the main technical lemma 
Lemma~\ref{lem:reg_technical} . In Section~\ref{sec:blowup} we prove 
the formal version of Theorem~\ref{thm:main_informal} 
(Theorem~\ref{thm-blup-main}) and deduce Corollary~\ref{cor:blup_edmonds} 
and~\ref{cor:blup_degree}.

\section{Preliminaries}\label{sec:prel}

\subsection{Notation}
\label{sec:notation}

For the reader's convenience, we collect the main notations in this section. Some 
of these were 
already introduced in the introduction. 

For $n\in\N$, 
$[n]:=\{1, \dots, n\}$. 
Given two vector spaces $U$ and $V$, $U \leq V$ denotes 
that $U$ is a subspace of $V$. $\zvec$ denotes the zero vector or the trivial 
vector space. 

Let $\F$ be a field. $\fdchar(\F)$ denotes the characteristic of $\F$. $M(n, \F)$ 
is the 
linear space of $n\times n$ matrices over $\F$. The rank of $A \in M(n,\F)$ is 
denoted $\rk(A)$. The corank of $A$, $\cork(A)$, is $n-\rk(A)$. For $U\subseteq 
\F^n$, $A^{-1}(U)=\{v\in \F^n \mid A(v)\in U\}$. $I$ denotes the 
identity matrix.

A linear subspace of $M(n, \F)$ is called 
a \emph{matrix space}.
For $B_i \in M(n, \F)$, $i\in[m]$, $\langle 
B_1, \ldots, B_m \rangle$  
denotes the matrix space spanned by the $B_i$'s. 
For a matrix space $\cB\leq M(n, \F)$, $\rk(\cB)$ is defined as 
$\max\{\rk(B)\mid B\in\cB\}$. We call $\cB$ \emph{singular}, if $\rk(\cB)<n$. For 
$\cB \leq M(n,\F)$ and $U\leq \F^n$, $\cB(U):=\langle \cup_{B\in\cB} B(U)\rangle$.
The non-commutative rank, 
$\nrk(\cB)$, is $n-c$ where $c$ is the maximum integer such that there exists a 
$c$-shrunk subspace. 

For $A\in M(d, \F)$ and $B\in M(n, \F)$, the tensor product $A \otimes B$ 
is a $d 
\times d$ block matrix with block size $n\times n$. 
For $A=(a_{i,j})_{i,j\in[d]}$, the $(i,j)$-th block 
of $A \otimes B$ is $a_{i,j}B$. 
For $\cB=\langle B_1, \dots, B_m\rangle\leq M(n, \F)$, the \emph{$d$th tensor 
blow-up},
$\blowup{\cB}{d}:=M(d, \F)\otimes \cB \leq 
M(dn,\F)$. A linear basis of $\blowup{\cB}{d}$ is $\{ E_{i,j}\otimes 
B_k\mid i, j\in[d], 
k\in[m]\}$, 
where $E_{i,j}$ is the matrix with $1$ at the $(i,j)$th position, and 
$0$ otherwise. 
In 
Section~\ref{sec:blowup} it will be easier to work with 
$\rblowup{\cB}{d}:=\cB\otimes M(d, \F)$. As $\blowup{\cB}{d}\cong 
\rblowup{\cB}{d}$, the latter will also 
referred to as the $d$th tensor blow-up of $\cB$. 

\subsection{The second Wong sequences}\label{subsec:wong}

Let us introduce a key tool to be used in Section~\ref{sec:blowup}, 
called the second\footnote{The first Wong sequence is the dual of the second one; 
this 
naming convention is due to Wong who in \cite{Wong} defined the two sequences for 
the special case $\cB=\langle B\rangle$.} (generalized) Wong sequence.  This was 
used by Fortin and Reutenauer 
\cite{FR04}, and rediscovered by the first two authors with Karpinski and Santha 
in \cite{conf_version} to solve Edmonds' problem for rank-$1$ spanned
matrix spaces over arbitrary fields. 

Given $A\in M(n, \F)$ and $\cB\leq M(n, \F)$, the second Wong sequence 
of 
$(A, \cB)$ is the following sequence of subspaces in $\F^n$: $W_0=\zvec$, 
$W_1=\cB(A^{-1}(W_0))$, \dots, $W_i=\cB(A^{-1}(W_{i-1}))$, \dots. It can be proved 
that $W_0<W_1<W_2<\dots<W_\ell=W_{\ell+1}=\dots$ for some $\ell\in\{0, 1, \dots, 
n\}$. $W_\ell$ is then called the limit of this sequence, denoted as $W^*$. 

A useful way to understand the second Wong sequence is to view it as a linear 
algebraic analogue of the augmenting path on bipartite graphs. While not precise, 
we find this intuition helpful. That is, we view 
matrices as linear maps from $V$ to $W$, $V\cong W\cong 
\F^n$. Vectors in $V$ and $W$ may be thought of as the ``vertices'' on the left 
and right part, respectively. Then for $A\in \cB$, thinking of $A$ as a given 
matching, $A^{-1}(\zvec)$ can be understood as identifying those ``vertices 
unmatched by $A$ on the left part.'' Then $\cB(A^{-1}(\zvec))$ is understood as 
taking those ``edges'' outside $A$, and $A^{-1}(\cB(A^{-1}(\zvec)))$ is understood 
as 
taking a further step with those ``edges'' in $A$. And so on. 

The key fact is that, when $A\in\cB$, $W^*\leq\im(A)$ if and only if there exists 
a $\cork(A)$-shrunk subspace \cite[Lemma 9]{conf_version} (reproduced below 
as Fact~\ref{fact:Wong}). If this is the 
case, 
$A$ 
is of maximum 
rank and $A^{-1}(W^*)$ is a $\cork(A)$-shrunk subspace. It is clear that
the second Wong sequence can be computed using polynomially many arithmetic 
operations. The direct way to  
compute the second Wong sequences over $\Q$ may cause the bit lengths to 
explode. If testing whether $W^*\leq\im(A)$ is the only concern (as in our 
application here), by replacing 
$A^{-1}$ with some appropriate ``pseudo-inverse'' of $A$, the bit lengths of the 
intermediate numbers up to the first $W_k$, $W_k\not\leq\im(A)$, can be 
bounded by a polynomial of the input size. We refer the reader to \cite[Lemma 
10]{conf_version} for this trick. 

When $|\F|$ is $\Omega(n)$, this immediately gives a 
method to decide whether $\nrk(\cB)=\rk(\cB)$ as in \cite{FR04}: randomly choose a 
matrix $A\in\cB$, which will be of maximal rank with high probability. Then 
compute the second Wong sequence of $(A, \cB)$ and check whether the limit 
$W^*\subseteq\im(A)$. 

For completeness we summarize the above discussion together as a fact. 

\begin{fact}[{\cite[Lemmas 9 and 10]{conf_version}}]
\label{fact:Wong}
Let $A \in \cB \leq M(n,{\F})$, and let $W^{*}$ be the limit of the second Wong 
sequence of
$(A, \cB)$. Then there exists a $\cork(A)$-shrunk subspace of $\cB$ if and only 
if $W^{*} \subseteq \im(A)$.
If this is the case then $A^{-1}(W^{*})$ is a 
$\cork(A)$-shrunk subspace of 
$\cB$. In the algebraic RAM model as well as over ${\mathbb Q}$ we can detect 
whether $W^*\subseteq \im(A)$ and if so compute a shrunk subspace
in deterministic polynomial time.
\end{fact}

For a matrix space $\cB$ of dimension $2$, $\rk(\cB)=\nrk(\cB)$ for large 
enough $\F$; this follows from the Kronecker-Weierstrass theory of matrix pencils, 
and alternative proofs may be found in \cite{EH88,PrimitiveI}. Due to this fact, 
it was observed in \cite{conf_version} that by utilizing the second Wong sequence 
we have the following. 
\begin{fact}[{\cite[Fact 11]{conf_version}}]\label{fact:dimtwo}
Assume that $|\F|>n$, and let $\cB=\langle A, B\rangle \leq M(n, \F)$. Then 
$\rk(A)=\rk(\cB)$ if and only if for any $i\in[n]$, $(\cB 
A^{-1})^i(\zvec)\leq  
\im(A)$. 
\end{fact}

\section{Gurvits' algorithm and Proposition~\ref{thm:sdit_main}}\label{sec:gurvits}

\paragraph{Commutative and the non-commutative Edmonds problem: a natural pair.} 
Viewing matrices as linear maps between two vector spaces, one may suspect Edmonds 
problem to be a linear algebraic analogue of the maximum 
matching problem on bipartite graphs, with elements of the
underlying vector spaces as being the left and right side vertices, and the 
matrices as giving us edges -- mapping a vector on the left side to one on the 
right side. Given such a correspondence, one may ask whether an analogue of Hall's 
theorem holds in this setting, i.e., is it true 
that a matrix space either has a matrix of rank $s$, or has an
$(n-s+1)$-shrunk subspace; or put differently, whether $\rk(\cB)=\nrk(\cB)$ holds 
for all $\cB$.
This is far from the truth! 
For example, for the space of skew-symmetric matrices 
($A=-\trans{A}$) 
of size $3$, we have $\nrk=3$ and $\rk=2$.

That is, while for the bipartite maximum matching problem, 
matchings and shrunk subsets are two sides of the same coin, in the linear 
algebraic setting, this coin splits into two problems:  Edmonds'
original (commutative) 
problem asks to compute the maximum rank, and the non-commutative 
Edmonds problem 
asks to compute the maximum $c$ for the existence of a $c$-shrunk 
subspace. 

\paragraph{Rank-1 spanned matrix spaces.} Now we point out that several 
results on the (commutative) Edmonds problem can be viewed, and should be 
understood as, 
resolving the 
non-commutative counterpart. For this, note that shrunk subspaces are a 
natural witness for the 
singularity of matrix spaces: this construction can be dated back to 1930's in T. 
G. Room's book \cite{Room}, and plays a key role in several results which solve 
special 
cases of Edmonds' problem including \cite{Lovasz,Gurvits,conf_version}. 

A particular 
case of interest is rank-$1$ spanned matrix spaces: those matrix spaces that have 
a basis consisting of rank-$1$ matrices. For rank-$1$ spanned spaces, the analogue 
of Hall's theorem holds \cite{Lovasz}, so the commutative and the non-commutative Edmonds 
problems coincide Therefore, the known results for rank-$1$ spanned spaces 
\cite{Gurvits,conf_version} can be viewed as solving either the non-commutative 
Edmonds problem or the commutative one. In retrospect, the results for rank-$1$ 
spanned spaces rely on 
shrunk subspaces in such a critical way that they should be understood as 
solving NCFullRank rather than the commutative 
version for this special case: 
\begin{itemize}
\item The core of Gurvits' algorithm \cite{Gurvits} is an iterative procedure 
called the operator 
Sinkhorn's scaling procedure. When applied to a matrix space $\cB$, this procedure 
converges, if and only if $\cB$ has a shrunk subspace.
\item The key tool in \cite{conf_version} is the second Wong sequence as described 
in Fact~\ref{fact:Wong}. When applied to $A\in \cB\leq M(n, \F)$, this 
sequence 
stabilizes in polynomially many number of steps, and the limit subspace is 
contained in $\im(A)$ if and only if $\cB$ has an $\cork(A)$-shrunk subspace. 
\end{itemize}

\paragraph{Gurvits' algorithm; Proof of Proposition~\ref{thm:sdit_main}.} In fact, 
Gurvits' algorithm works by assuming that an analogue 
of Hall's theorem for perfect matchings holds. 
\begin{theorem}[\cite{Gurvits}]\label{thm:gurvits}
Over $\Q$, given a matrix space $\cB$ such that either $\rk(\cB)=n$ or $\nrk(\cB)<n$, 
there exists a deterministic polynomial-time algorithm that solves SDIT, and 
therefore NCFullRank. 
\end{theorem}

Gurvits' algorithm almost solves NCFullRank over $\Q$. The only 
problem is that for a matrix space $\cB$ with $n=\nrk(\cB)>\rk(\cB)$ 
the algorithm may give a 
wrong answer. 
(We note that even when the input to the algorithm is such a matrix space it 
terminates 
in polynomially many steps.) We observe that this can be rectified by considering 
matrix semi-invariants up to the upper bound for $\sigma(R(n, m))$. 

\begin{proof}[Proof of Proposition~\ref{thm:sdit_main}]
Recall that, by assumption, the nullcone of $R(n, m)$ is defined by elements of 
degree 
$\leq \sigma(R(n, m))$ over $\Q$. Also, given a matrix space 
$\cB\in M(s, \Q)$, Gurvits' algorithm either reports that $\rk(\cB)=s$, or 
$\nrk(\cB)<s$. When $\rk(\cB)=s$ or $\nrk(\cB)<s$, it is always correct.

The algorithm is easy to describe: for $d=1, \dots, \sigma = \sigma(R(n, m))$ run 
Gurvits' algorithm with input
$\blowup{\cB}{d}$. If for some $d$, Gurvits' 
algorithm reports $\rk(\blowup{\cB}{d})=dn$, then 
output $\nrk(\cB)=n$ and halt. Otherwise, return $\nrk(\cB)<n$. 

It is clear that this algorithm runs in time polynomial in the input size and 
$\sigma$. Note that a linear basis of $\blowup{\cB}{d}$ can be constructed easily 
in time polynomial in the input size of $\cB$ and $d$. 

From the discussion in Section~\ref{subsec:inv}, the correctness is also 
easy to see. Specifically, assuming the bound on $\sigma$, the simultaneous 
vanishing of
$
\det(Y_1\otimes B_1+\dots+Y_m\otimes B_m) \in \Q[y^{(k)}_{i,j}]
$
for $d=1, \dots, \sigma$, characterizes whether $\nrk(\cB)<n$ or not. Therefore, 
if $\nrk(\cB)=n$, then for some $d\leq \sigma$, $\rk(\blowup{\cB}{d})$ is full. On 
the other hand if $\nrk(\cB)<n$, then there is a shrunk subspace $U$. For each 
$d$, 
$\Q^d \otimes U$ is a shrunk subspace of $\blowup{\cB}{d}$ and so  
$\nrk(\blowup{\cB}{d})<dn$ for any $d$.
\end{proof}

\paragraph{Implications to Gurvits' algorithm.} The invariant-theoretic viewpoint 
also connects to a question of Gurvits in \cite{Gurvits}. In \cite{Gurvits}, given 
a basis $\{B_1, \dots, B_m\}$ of  $\cB\leq M(n, \C)$, Gurvits 
associates with it a completely positive operator i.e., a linear map $F:M(n, 
\C)\to M(n, \C)$. The main algorithmic technique is the so-called operator 
Sinkhorn's iterative scaling procedure, which is applied to $F$. This procedure is 
a quantum generalization of the classical Sinkhorn's iterative scaling procedure, 
which is applied to nonnegative matrices, and can be used to approximate the 
permanent, and to decide the existence of perfect matchings \cite{GY98,LSW00}. 
Gurvits proved 
that this procedure, when applied to the operator $T$ derived from a matrix space 
$\cB$, converges if and only if $\cB$ has a shrunk subspace. He proved this using 
a continuous but non-differentiable function, called the capacity of an operator, 
denoted as $\gcap(F)$. Specifically, he showed that $\cB$ has a shrunk subspace if 
and only if $\gcap(F)=0$. 
Gurvits asked whether there exists a ``nice'' function, like a polynomial with 
integer coefficients, that characterizes $\cB$ with shrunk subspaces. Our previous 
argument suggests that there exists a set of polynomial functions with integer 
coefficients, whose simultaneous vanishing characterizes those $\cB$ with shrunk 
subspaces, and therefore a ``nice'' substitute for Gurvits' capacity. However, the 
number of these polynomial functions depends on the degree 
bound for matrix semi-invariants.

\section{Efficient construction of division algebras}\label{sec:cda}


Division algebras, and efficient construction 
of such algebras with explicit matrix representations play a 
crucial role in the main technical lemma, Lemma~\ref{lem_reg_blowup}, in this 
paper. In this section we 
present an efficient construction of such algebras based on Kummer extensions. 


\subsection{Basic facts about central division algebras} 

Let us first introduce 
some basic facts about central division algebras. Proofs of 
the these statements can be found in ~\cite{Lam}. 

Let $\F$ be a field. A {\em division algebra} $D$ over $\F$ is an
associative $\F$-algebra in which the non-zero elements
are invertible. The {\em center} of a division algebra $D$
over $\F$ is obviously an extension field of $\F$. All the division
algebras considered in this section are finite
dimensional over their center. The {\em opposite division algebra} $D^{op}$ is the 
algebra with the same set of 
elements as $D$ and with multiplication $x \cdot y$ defined to be $y*x$, with $*$ 
being the multiplication in $D$. 

When
the center coincides with $\F$, we say that $D$ is
a {\em central division algebra} over $\F$, and in this case, 
$\dim_{\F}(D)= d^2$ for some positive integer $d$.
This $d$ is called the {\em index} of $D$. 
For $x,y\in D$ we can consider the linear transformation 
$\mu_{x,y}$ on $D$, considered as a vector space of dimension $d^2$ over
$\F$, defined as $\mu_{x,y}z=x*z*y$. The linear extension of the map
$x\otimes y\mapsto \mu_{x,y}$ to $D\otimes D^{op}$ gives an 
isomorphism $D \otimes D^{op} \cong 
M(d^2, \F)$, the algebra of $d^2 \times d^2$ matrices 
with entries in $\F$. 
The image of $D \otimes I$  under this isomorphism gives a representation of $D$ in 
the space of $d^2 \times d^2$ matrices with entries in $\F$. 
What is important for us is the observation that 
matrices giving us a 
representation of $D^{op}$ in $M(d^2, \F)$ commute with the matrices giving us a 
representation of $D$ in $M(d^2,\F)$.

\subsection{Constructing cyclic field extensions under a coprime 
condition}\label{subsec:cyclic}

Our division algebras will be cyclic algebras, that is, non-commutative algebras
constructed from cyclic extensions of fields. In this subsection
we present an efficient construction of such field extensions, 
under the condition that the extension 
degree and the field characteristic are coprime. 

Recall that a {\em cyclic extension of a field} $\K$
is a finite Galois extension of $\K$ having a cyclic Galois group.
By constructing a cyclic extension $\L$ we mean constructing
the extension as an algebra over $\K$, e.g., by giving an
array of {\em structure constants} with respect to a $\K$-basis for $\L$
defining the multiplication on $\L$ as well as specifying a generator 
of the Galois group, e.g, by its matrix with respect to a $\K$-basis.
Recall that for a finite dimensional algebra $\cA$ over the field
$\K$, a common way to specify the multiplication is using an array
of structure constants with respect to a $\K$-basis $A_1,\ldots,A_d$.
These are $d^3$ elements $\gamma_{ijk}$ of $\K$
such that $A_iA_j=\sum_{k=1}^d\gamma_{ijk}A_k$. 
Then we can represent elements of $\cA$  by the vectors of
their coordinates in terms of the basis $A_1,\ldots,A_d$.
The size of the data representing the structure constants gives
some control over the size of the data representing the product
of elements. For example, consider the following situation: 
$\K$ is the function field $\F'(Z)$, where $\F'$ is a field and $Z$ a formal 
variable. The structure constants
happen to be polynomials in $\F'[Z]$. Then for two elements
of $\cA$ with their coordinates being polynomials in $\F'[Z]$, their product
will have also polynomial coordinates, and the degrees
of the coordinates of the product are upper bounded by the
sum of the maximum degrees of coordinates of the factors, plus the 
maximum degree of the structure constants. 

\begin{lemma}
\label{lem:cyclic_gen}
Let $\F'$ be a field. Let $d$ be any non-negative integer if the characteristic of 
$\F'$ is zero, otherwise assume that $d$ is not divisible by the characteristic of 
$\F'$.  Assume that $\F'$
contains a known primitive
$d$th root of unity $\zeta$, and let $X$ be a formal variable. 
Then 
a cyclic extension $\L$ having degree $d$ over $\K:=\F'(X)$ can 
be computed using $\poly(d)$ arithmetic operations. $\L$ will be given by structure
constants with respect to a basis, and the matrix for a generator of
the Galois group of $\L/\K$ in terms of the same basis will 
also be given.
All the output entries (the structure constants as well as the entries 
of the matrix representing the Galois group generator) will be
polynomials of degree $\poly(d)$ in $\F'[X]$.
Furthermore for $\F'=\Q[\sqrt[d]{1}]$, the bit complexity 
of the algorithm (as well as the size of the output) is
$\poly(d)$. 
\end{lemma}

\begin{proof}
Put $\L=\F'(Y_1)$ where $X=Y_1^{d}$.
Then $1,Y_1,\ldots,Y_1^{d-1}$ are a $\F'(X)$-basis for
$\L$ with $Y_1^iY_1^j=Y_1^{i+j}$ if $i+j\leq d$ and
$XY_1^{i+j-d}$ otherwise. Further note that the linear extension $\sigma$ of the 
map sending $Y_1^j$ to $\zeta^jY_1^j$ is a $\K(X)$-automorphism of degree
$d$.
\end{proof}

\begin{remark}\label{remark:coprime}
The construction above is known in the literature as a Kummer extension.
When the characteristic is a prime $p$ and a divisor of $d$, say
$d=p^ed'$ where $d'$ is prime to $p$, the Kummer
extension $\K$ should be replaced by a cyclic extension
which is a product of a Kummer extension of degree $d'$ 
and a cyclic extension of degree $p^e$ described by 
Artin, Schreier and Witt \cite{Witt37}. Investigating
 the complexity of computing such extensions requires some 
further work. In 
\cite{IQSnote}, we conduct such a research and present an efficient construction 
of such extensions. The consequence on results in this paper will be reported in 
ibid..
\end{remark}

\subsection{Constructing cyclic division algebras}

The following statement connects cyclic field extensions with central division 
algebras. It follows from Wedderburn's theorem
characterizing cyclic division algebras
(see e.g. \cite[Theorem~(14.9)]{Lam}) as shown
on Page 221 of~\cite{Lam}.

\begin{fact}
\label{fact:cyclic_algebra}
Let $\L$ be a cyclic extension of degree $d$ of a field $\K$.
Let $\sigma$ be a generator of the Galois group, and $Z$ a number transcendental
over $\L$. 
For the transcendental extension $\L(Z)$ of $\L$, $\sigma$
extends to an automorphism (denoted again by $\sigma$)
of $\L(Z)$ such that the fixed field of $\sigma$ is $\K(Z)$.
Thus $\L(Z)$ is a cyclic extension of $\K(Z)$. Consider
the $\K(Z)$-algebra $D$ generated by (a basis for) $\L$ 
and by an element $U$ with relations $U^d=Z$ and
$Ua=a^\sigma U$
(for every $a \in \L(Z)$, or, equivalently for every $a$ from
a fixed $\K$-basis for $\L$).
Then $D$ is a central division algebra of index $d$ over
$\K(Z)$.
\end{fact}

The following proposition is an algorithmic realization of 
Fact~\ref{fact:cyclic_algebra}. 

\begin{proposition}\label{prop:cyclic_algebra}
Let $\L$ be a cyclic extension of degree $d$ of a field $\K$, and suppose that 
$\L$ is given by structure constants with respect to a $\K$-basis $A_1, \ldots, 
A_d$. 
Similarly, a generator $\sigma$ for the Galois group is assumed to be given
by its matrix in terms of the same basis.
Let 
$Y$ be 
a formal variable.
Then one can construct a $\K(Y)$-basis
$\Gamma$ of $M(d,\K(Y))$  such that
the $\K(Y^d)$-linear span of $\Gamma$
is a central division algebra over $\K(Y^d)$
of index $d$, using $\poly(d)$ arithmetic operations in $\K$.
\end{proposition}
\begin{proof}
Let $Z=Y^d$. Let $D$ 
be a central division algebra over $\K(Z)$ 
as in Fact~\ref{fact:cyclic_algebra}. 
The existence of a $\K(Z)$-subalgebra 
$D'$ of $M(d,\K(Y))$
isomorphic to $D$ follows, e.g., from 
Theorem~(14.7) of~\cite{Lam}. 
To construct a basis $\Gamma$ for such a
matrix algebra $D'$ efficiently, note that $A_iU^j$, $i, j=1, \ldots, d$, form a 
$\K(Z)$-basis of $D$.
This is also a $\K(Y)$-basis
for the algebra $D''=\K(Y)\otimes_{\K(Z)} D$. Consider also
the element $U_0=\frac{1}{Y} \otimes U$. 
Then $U_0^d=1$.
As the elements
$U_0^j$ are linearly independent 
over $\K(Y)$ and hence over $\K(Z)$ as well, 
we have that $E=U_0+U_0^2+...+U_0^{d-1}+U_0^d$ 
is nonzero. As $U^jE=Y^jE$, we have $A_iE$ ($i=1,\ldots,d$) 
form a $\K(Y)$-basis for the left ideal $D''E$ of dimension $d$. 
Now the action of $D''$ on this left ideal gives
a matrix representation for $D''$. 
Let $\gamma_{kij}$ be the structure constants
for the multiplication of $\L$ (and of $\L(Z)$):
$$A_kA_i=\sum_{j=1}^d\gamma_{kij}A_j\;\;(k,i=1,\ldots,d).$$
Also, let $\delta_{\ell ij}$ be the entries of the matrix
of the $\ell $th power of the generator $\sigma$ of the Galois group:
$$A_i^{\sigma^\ell}=\sum_{j=1}^d\delta_{\ell ij}A_j\;\;(\ell,i=1,\ldots,d).$$
(Notice that the matrix $(\delta_{\ell ij})_{ij}$ is
the $\ell$th power of $(\delta_{1 ij})_{ij}$, whence the degrees 
of its elements are 
also bounded by $\poly(d)$.)
Then 
$$A_kA_iE=\sum_{j=1}^d \gamma_{kij}A_jE$$
and 
$$U^\ell A_iE=A_i^{\sigma^\ell} U^\ell E=A_i^{\sigma^\ell} Y^\ell E=
Y^\ell \sum_{j=1}^d\delta_{\ell ij}A_jE.$$
Thus the matrix of the action of $A_k$ has entries 
$\gamma_{k ij}$ and the matrix of the action of $U^\ell$ has
entries $Y^\ell\delta_{\ell ij}$. Then the action
of $U^\ell A^k$ can be obtained as the product of these two
matrices. Let $\Gamma$ consist of all such $d^2$ products, and the proof is 
concluded.
\end{proof}

Combining Lemma~\ref{lem:cyclic_gen} and Proposition~\ref{prop:cyclic_algebra}, we 
immediately obtain the following. 

\begin{lemma}
\label{lem:cyclic_splitting}
Let 
$d$, $X$, $\F'$, $\K=\F'(X)$ and $\L$ be as in 
Lemma~\ref{lem:cyclic_gen}. In particular, if $\fdchar(\K)=p>0$ then $p\nmid d$. 
Then one can construct a $\F'(X,Y)$-basis
$\Gamma$ of $M(d,\F'(X,Y))$  such that
the $\F'(X,Y^d)$-linear span of $\Gamma$
is a central division algebra over $\F'(X,Y^d)$
of index $d$, using $\poly(d)$ arithmetic operations in $\F'$.
In particular, for any $A\in \Gamma$, the 
entries of $A$ are polynomials in $\F'[X, Y]$ of degree $\poly(d)$.
Furthermore for $\F'=\Q[\sqrt[d]{1}]$, the bit complexity 
of the algorithm (as well as the size of the output) is
also $\poly(d)$. 
\end{lemma}

\subsection{Some algorithmic issues for actual applications} To put the above 
construction in action, we need to handle a few algorithmic problems as follows.

\subsubsection{Algorithmic issues when working with field extensions} 
\label{subsubsec:zeta}

Lemma~\ref{lem:cyclic_gen} assumes the field $\F'$ contains a 
known primitive $d$th root of unity 
$\zeta$, where if $\fdchar(\F')=p>0$ then $p\nmid d$. In actual applications, we 
may 
start with a field $\F$ without a primitive 
$d$th root of unity in it, and attach one symbolically,
which we still denote by 
$\zeta$. However, this 
may cause some problem. Namely, constructing 
$\F'=\F[\zeta]$ would require factoring 
the polynomial $x^{d}-1$ over $\F$, a task which cannot be 
accomplished using basic arithmetic operations. To see that
this is indeed an issue notice that
 a black-box field may contain certain ``hidden'' parts
of cyclotomic fields. 
Of course,
over certain concrete fields, such as the rationals,
number fields or finite
fields of small characteristics, this can be done in polynomial time.
However, even over finite fields of large characteristic no
deterministic polynomial time solution to this task is known
at present.

To get around this issue, one can
perform the required computations over an appropriate factor algebra $R$ of
the algebra $C=\F[x]/(x^{d}-1)$ in place $\F'$ as if $R$ were a field. 
To be specific, as $d$ is not divisible by the characteristic,
we know that $C$ is semisimple -- actually it is isomorphic to
a direct sum of ideals, each of which is isomorphic
to the splitting field $\F[\sqrt[e]{1}]$ of the polynomial $x^{e}-1$ for
some divisor $e$ of $d$, and the projection of $x$ to such an ideal
is a primitive $e$th root of unity. It follows that if we compute the ideal $J$ generated by 
annihilators of $x^e -1$, for all $e$ a proper divisor of $d$, then $R=C/J$ is isomorphic to the direct
sum of copies of the splitting field $\F'$ of $x^{d}-1$, and the projection
of $x$ to each component is a primitive $d$th root of unity. And this
property is inherited by any proper factor of $R$. A
computation using $R$ instead of $\F'$ may fail
only at a point where we attempt to invert an non-invertible
element of $R$. However, such an element must be a zero divisor.
When this situation occurs, we replace $R$
with the factor of $R$ by its ideal generated by the zero divisor and
restart the computation. Such a restart can clearly happen at most $d-2$ times.

We explain what the above scheme entails in our actual tasks. 

As the methods 
for Proposition~\ref{prop:cyclic_algebra} and Lemma~\ref{lem:cyclic_splitting} 
do not require division, the zero divisor issue does not occur there. Replacing 
$\F'$ with $R$, the outcome
$\Gamma$ of Lemma~\ref{lem:cyclic_splitting} will actually be a free
 $R(X,Y)$-basis for an algebra which
is a direct sum of isomorphic copies of a division algebra,
embedded into $M(d,R(X,Y))$. 

Now consider the task of computing the rank of $M(N, \F')$. Note that we cannot 
talk about the ``rank'' of matrices in
$M(N, R)$ which is not well-defined. But since $R$ is a direct sum of $\F'$, the 
decomposition of $R$ induces a decomposition of $M(N, R)$ into a direct sum of 
copies of $M(N, \F')$. We call the images of the projections of a matrix
$B\in M(N, R)$ to the direct summands 
the {\em components} of $B$. The following lemma describes how to compute the 
maximum rank over the components. 
\begin{lemma}
\label{lem:correct_rank}
Let $R$ and $\F'$ be as above, and suppose we are given a matrix $B\in M(N, R)$. 
Then there exists a deterministic polynomial-time algorithm that computes the 
maximum rank over the components of $B$. 
\end{lemma}
\begin{proof}
This can be achieved by combining division-free algorithms for computing the 
determinant by e.g. Kaltofen \cite{Kal92} (see also \cite{maha-vinay} for more 
such algorithms), and the parallel algorithm for computing the rank of a matrix by 
Mulmuley 
\cite{Mul87}. 

We include a sketch here for completeness. To start with, instead of 
$B$ we consider the symmetric matrix $B'=\begin{pmatrix} 
0 & B \\ \trans{B} & 0 \\ \end{pmatrix}$. Then let $x$ and $y$ be two formal 
variables. 
Form a matrix $D=\mathrm{diag}(1, y, y^2, \dots, y^{2N-1})$, and compute 
$\det(xI-DB')$ using \cite{Kal92}, 
considered as a polynomial in $\F'(y)[x]$. Let $M$ be the maximum integer such 
that 
$x^M$ divides $\det(xI-DB')$, and return $(2N-M)/2$. 

By \cite{Mul87}, the above 
procedure on a matrix from $M(N, \F')$ returns its rank. Now for $B\in M(N, R)$, 
since it is (implicitly) a direct sum of several copies of $M(N, \F')$, the above 
algorithm on 
$B$ can be viewed as working with these components ``in parallel'', and the 
resulting $\det(xI-DB')$ is a direct sum of $\det(xI-DB_i')$ where $B_i'$ are the 
components of $B'$. It is then not hard to deduce that the above procedure 
computes the maximum rank over the components of $B$. 
\end{proof}
\begin{remark}
Using the method of Lemma~\ref{lem:correct_rank} for rank computations,
we will obtain an algorithm that does not require division in $R$ at all, 
and hence we will not need the above mentioned restarts.  
Another possibility would be doing Gaussian elimination and 
restarting the computation once a 
zero divisor is met as described. If no zero divisors are met and the rank is $r$, 
then it means that the columns of $B$ generated a free module over $R$ of rank 
$r$, so each component is also of rank $r$ over $\F'$. 
\end{remark}

Finally, we note that a similar issue, namely that a black box field may even 
contain infinite algebraic extensions of its subfields has been circumvented
by using the transcendental extension $\K=\F'(X)$
in the construction of cyclic extensions (Lemma~\ref{lem:cyclic_gen}).

\subsubsection{Computing the rank of matrices over a rational function field in 
few variables} 

Note that the matrices from Lemma~\ref{lem:cyclic_splitting} are matrices over a 
rational function field. Therefore we will need to compute the rank of matrices in 
such form. 

\begin{proposition}\label{prop:rank}
Let $\F'$ be a field and $\K=\F'(X_1, X_2, \dots, X_k)$ be a pure transcendental 
extension of $\F'$. Let $A$ be an $N \times N$ matrix with entries as quotients of 
polynomials from $\F'[X_1, X_2, \dots, X_k]$, where the polynomials are explicitly 
given as sums of monomials. Assume that the degrees of the 
polynomials appearing in $A$ are upper bounded by $D$. If 
$|\F'|=(ND)^{\Omega(k)}$, then we can find in 
time $(ND)^{O(k)}$ a matrix $B\in M(N, \F')$ with $\rk(B)=\rk(A)$.
\end{proposition}
In particular, if $k$ is a constant -- $k=2$ as used in Lemma~\ref{lem:red_data} 
for the procedure in Lemma~\ref{lem_reg_blowup} 
-- then the above procedure runs in polynomial time. 

\begin{proof}
We multiply the entries of $A$ by an easily computable common multiple (e.g., the 
product) of their denominators to 
obtain a matrix with polynomial entries from $\F'[X_1, X_2, \dots, X_k]$. The data 
describing
this matrix has size polynomial in the size of the input data. In particular, the 
degree 
of the determinant of any sub-matrix is upper bounded by a polynomial $s$ in 
$(ND)^{k}$. We have assumed 
that $|\F'|=(ND)^{\Omega(k)}$. Then from $(ND)^{O(k)}$ 
specializations by 
elements of a subset
of size $s+1$ of $\F'$, at least one gives a matrix with entries
$\F'$ having the same rank as the original matrix. Thus the rank of
$A$ can be computed by computing the rank of $(ND)^{O(k)}$ matrices
over $\F'$. 
\end{proof}
Note that if we simulate $\F'$ using $R$ as described in 
Section~\ref{subsubsec:zeta} then we shall apply the procedure in 
Lemma~\ref{lem:correct_rank} after specializing the variables as in 
Proposition~\ref{prop:rank}.

\section{Finding a nonsingular matrix in blow-ups}\label{sec:blowup}

In this section we describe, given $\cB\leq M(n, \F)$, how to compute a 
nonsingular matrix in $\blowup{\cB}{d}=M(d, \F)\otimes \cB$ for some $d\leq 
(n+1)!$, or certify that this is not possible. 

One important note is due here: as per our notation, matrices in $M(d, \F)\otimes 
M(n, \F)$ are viewed 
as block matrices, where each block is of size $n\times n$. This is more 
convenient when describing semi-invariants. In this section, it will be more 
convenient to work with $M(n, \F)\otimes M(d, \F)$, namely each block is of size 
$d\times d$. This is 
consistent with other parts simply because $M(d, \F)\otimes M(n, \F)\cong M(n, 
\F)\otimes M(d, \F)$. Therefore, we identify $M(nd, \F)\cong M(n, \F)\otimes M(d, 
\F)$, and fix such a decomposition. Recall the notation $\cB(U)$, and 
$\rblowup{\cB}{d}$ from 
Section~\ref{sec:notation}.

After some 
preparation, we prove the main technical lemma called the regularity lemma for  
blow-ups. Then we prove Theorem~\ref{thm-blup-main} (the formal version of 
Theorem~\ref{thm:main_informal}), and Corollaries~\ref{cor:blup_edmonds} 
and~\ref{cor:blup_degree} follow easily.

\subsection{Preparations}
\label{subsec:prep}

\paragraph{A characterization of blow-ups.}
\begin{proposition}
For $\cA\leq M(dn, \F)$, $\cA=\rblowup{\cB}{d}$ for some $\cB \leq M(n, \F)$
if and only if $(I\otimes M(d, \F))\cA(I\otimes M(d, \F))=\cA$. 
\end{proposition}
We remind the reader that $(I\otimes M(d, \F))\cA(I\otimes M(d, \F))=\{XAY : X\in 
I\otimes M(d, \F), A\in \cA, Y\in I\otimes M(d, \F)\}$. 
Note that multiplication of tensor products of matrices obeys the rule $(X_1 
\otimes Y_1)(X_2 \otimes Y_2) = X_1 X_2 \otimes Y_1 Y_2$.

\begin{proof}
The only if part is obvious. To see the reverse implication, for $i,j\in [d]$, 
let $E_{ij}$ stand for the elementary matrix in $M(d,\F)$ in
which the $(i, j)$th entry is $1$ and the
others are zero. Then for every quadruple
$(i,j,i',j')\in [d]^4$ we have 
$\sum_{k=1}^{d}E_{ki}E_{i'j'}E_{jk}=\delta_{ii'}\delta_{jj'}I$,
where $\delta_{\ell \ell'}$ stands for the Kronecker delta.
Any element $A$ of 
$\cA$ can be written as 
$\sum_{i,j=1}^dB_{ij}\otimes E_{ij}$ with
$B_{ij}\in M(n,\F)$. Then for every $i,j\in [d]$ we have
$\sum_{k=1}^d(I\otimes E_{ki})A(I\otimes E_{jk})=B_{ij}\otimes I$,
which implies that $B_{ij}\otimes I$ are in $\cA$.
So define $\cB$ as $\{B\in M(n, \F): \text{ such that } 
B\otimes I\in \cA\}$, and we see that 
$\cA=\rblowup{\cB}{d}$.
\end{proof}
Note that  $(I\otimes M(d, \F))\cA(I\otimes M(d, \F))=\cA$ is equivalent to 
saying 
that $\cA$ is an $M(d, \F)$ sub-bimodule of $M(n, \F)\otimes M(d, \F)$, where we 
identify $M(d, \F)$ with $I\otimes M(d, \F)$. 
Similarly, for a subspace $W\leq \F^n\otimes \F^d$, one can see
that $W$ is of the form $W_0\otimes \F^d$ if and only if 
$(I\otimes M(d, \F))W_0=W_0$, that is, $W_0$ is an $M(d, \F)$-submodule of
$\F^n\otimes \F^d$.

\paragraph{Shrunk subspaces in the blow-up situation.}
\begin{proposition}
If $\cA=\rblowup{\cB}{d}$ has an $s$-shrunk subspace, then $\cA$ has an 
$s'$-shrunk subspace where $s'\geq s$ such that $d$ divides $s'$, and $\cB$ has an 
$s'/d$-shrunk subspace. 
\end{proposition}
\begin{proof}
As $\cA=\rblowup{\cB}{d}$,
$(I\otimes M(d, \F))\cA(I\otimes M(d, \F))=\cA$. Assume that $U$ is 
an $s$-shrunk subspace of $\cA$: with $W=\cA(U)$ we have
$\dim_\F U-\dim_\F W=s$. Then $(I\otimes M(d, \F))W=(I\otimes M(d, \F))\cA U=\cA 
U=W$,
thus $W=W_0\otimes \F^d$ for some $W_0\leq \F^n$. Similarly, 
as $\cA(I\otimes M(d, \F))=\cA$, we have
$\cA(I\otimes M(d, \F))U=W$, whence $U'=(I\otimes M(d, \F))U$ is an
$s'$-shrunk subspaces with $s'\geq s$, and $U'=U_0\otimes \F^d$
with some $U_0\leq \F^n$. Note that $\dim(W)$, $\dim(U')$, and therefore $s'$, are 
all divisible by $d$. Noting $\cA=\cB\otimes M(d, \F)$, we have
$W_0\leq \cB(U_0)$ and so $U_0$ is an $s'/d$-shrunk subspace.
\end{proof}

\paragraph{From the extension field to the original field.}
Assume that for some extension field $\K$ of $\F$ we are given a
matrix $A'\in \cB \otimes_\F \K\leq M(n, \K)$ of rank $r$. Then, if $|\F|> r$,
using the method of \cite[Lemma 2.2]{GIR}, we can efficiently find a matrix $A\in 
\cB$ 
of rank at least $r$. This procedure is also useful to keep sizes of
the occurring field elements small. For completeness we include a brief 
description. Let $S\subseteq \F$ with $|S|=r+1$ and
let $B_1,\ldots,B_\ell$ be an $\F$-basis for $\cB$. Then 
$A'=a_1'B_1+\ldots+a_\ell'B_\ell$, where $a_i'\in\K$. As $A'$ is of rank $r$, 
there 
exists an $r\times r$ sub-matrix of $A$ with nonzero determinant. Assume that 
$a_1'\not \in S$. Then
we consider the determinant of the corresponding sub-matrix of the polynomial 
matrix $xB_1+a_2'B_2+\ldots a_{\ell}'B_\ell$. This determinant is
a nonzero polynomial of degree at most $r$ in $x$. Therefore there exists
an element $a_1\in S$ such that $a_1B_1+a_2'B_2+\ldots a_{\ell}'B_\ell$ has rank
at least $r$. Continuing with $a_2',\ldots,a_{\ell}'$, we can ensure
that all the $a_i$'s are from $S$. Since the $B_i$'s span $\cB$, the resulting 
matrix of rank at least $r$ is
in $\cB$. We record this too as a fact.

\begin{lemma}[{Data reduction, \cite[Lemma 2.2]{GIR}}]
\label{lem:red_data}
Let $\cB\leq M(k\times \ell,\F)$ be given by a basis
$B_1,\ldots,B_m$, and let $\K$ be an extension field of $\F$. Let $S$ be a subset 
of $\F$
of size at least $r+1$. Suppose that we are 
given a
matrix $A'=\sum a_i'B_i \in \cB\otimes_{\F}\K$ of rank at least $r$. 
Then 
we can find $A=\sum a_i B_i\in \cB$ of rank also at least $r$
with $a_i\in S$. The algorithm uses $\poly(k, \ell, r)$
rank computations for matrices of the form $\sum a_i''B_i$
where $a_i''\in \{a_1',\ldots,a_m'\}\cup S$.  
\end{lemma}

\subsection{Regularity of blow-ups}

Our goal in this subsection is to 
prove that when the field size is large enough, the maximum rank over 
$\cA=\rblowup{\cB}{d}\leq M(dn, \F)$ is 
always divisible by $d$. The proof is constructive when $\fdchar(\F)=0$, or when
$\fdchar(\F)\nmid d$: if we get a 
matrix in $\cA$ of rank at least $rd+1$, we will be able to construct a matrix in 
$\cA$ of 
rank at least $(r+1)d$. This is the main technical tool to be used in the proof 
of Theorem~\ref{thm-blup-main}.

We first present a version of the regularity lemma in which a
matrix division algebra as in Lemma~\ref{lem:cyclic_splitting}
is assumed to be part of the input. 

\begin{lemma}[Regularity of blow-ups, technical version]
\label{lem:reg_technical}
Assume that we are given a matrix 
$A\in \rblowup{\cB}{d}\leq M(dn, \F)$ with $\rk(A)=(r-1)d+k$ for some
$1<k<d$. Let $X$ and $Y$ be formal variables and put
$\K=\F'(X)$, where $\F'$ is a finite extension of $\F$ of degree at most $d$. 
Suppose further that $|\F|>(nd)^{\Omega(1)}$ and 
that
we are also given a $\K(Y)$-basis $\Gamma$ of $M(d,\K(Y))$  such that
the $\K(Y^d)$-linear span of $\Gamma$
is a central division algebra $D'$ over $\K(Y^d)$. Let $\delta$
be the maximum of the degrees of the polynomials appearing as
numerators or denominators of the entries of the matrices 
in $\Gamma$. Then, using $(nd+\delta)^{O(1)}$ arithmetic
operations in $\F$, one can find a matrix $A''\in \rblowup{\cB}{d}$ with 
$\rk(A'')\geq rd$. Furthermore, over $\Q$ the bit complexity of
the algorithm is polynomial in the size of the input data (that is,
the total number of bits describing the entries of matrices and the
coefficients of polynomials).
\end{lemma}
\begin{proof}
Instead of the blow-up $\rblowup{\cB}{d,d}
=\cB\otimes M(d, \F)$
we start with the blow-up $\cB\otimes M(d,\K(Y))$. Then both 
$\rblowup{\cB}{d,d}$ and $\cB\otimes_\F D'$ span, as a $\K(Y)$-linear space,
the blow-up $\cB\otimes_\F M(d, \K(Y))$.

\begin{claim}
\label{claim:full_rank}
Every matrix in 
$M(n,\F)\otimes D'\subset M(d, \K(Y))$ has rank (as a matrix over $\K(Y)$) 
divisible by $d$.
\end{claim}

\begin{proof}
Firstly note that $M(n, \F)\otimes_\F D'=M(n, \K(Z))\otimes_{\K(Z)} D'$, since 
$D'$ is a $\K(Z)$-algebra. As $D'\subset M(d, \K(Y))$, $M(n, \K(Z))\otimes_{\K(Z)} 
D'$ acts naturally on the $\K(Z)$-space $\K(Z)^n\otimes \K(Y)^d\cong \K(Z)^n\otimes 
\K(Z)^{d^2}\cong \K(Z)^{nd^2}$. 
Since $D' \otimes_{\K(Z)} D'^{op} \cong M(d^2, 
\K(Z))$~\cite[Corollary 15.5]{Lam}, it follows that the centralizer of this 
action is isomorphic to the opposite algebra ${D'}^{op}$. 
Therefore the image $A'\K(Y)^{dn}$ of any $A'\in M(n, \F)\otimes_{\F} 
D'$ is a ${D'}^{op}$-submodule, whence its dimension over $\K(Z)$ is divisible by 
$d^2$. It follows that the dimension over $\K(Y)$ is divisible by $d$.
\end{proof}

The claim enables us to ``round up'' the rank of $A$ to the next multiple of $d$. 
Let $B_1, \dots, B_d$ be an $\F$-basis of $\cB$. 
Since $\rk(A)>(r-1)d$ over $\F$, clearly $\rk(A)>(r-1)d$ over $\K(Y)$ as well. 
 Now $\Gamma$, the $\K(Z)$-basis for $D'$ is a $\K(Y)$-basis
for $M(d,\K(Y))$.
Therefore $A$, as a matrix over $\K(Y)$, can be expressed as a linear combination (with 
coefficients over $\K(Y)$)
of $\{B_i\otimes C: i\in [d], C\in \Gamma\}$. We use the method of 
Lemma~\ref{lem:red_data} to
find coefficients from $\K(Z)$ (or even from $\F$)
such that the combination $A'$ of the basis element for $D'$
has rank also larger than $(r-1)d$. We have $A'\in \cB\otimes D'$, whence
by Claim~\ref{claim:full_rank}, the rank of $A'$ is at least
$rd$. Then we express $A'$ as 
a linear combination of elements -- with coefficients from
$\K(Y)$ -- of an $\F$-basis of $\rblowup{\cB}{d,d}$ which
is also a $\K(Y)$-basis for $\cB\otimes M(d,\K(Y))$. Then we use
again the algorithm of Lemma~\ref{lem:red_data} to replace
these coefficients to elements of $\F$ to find a matrix
$A''\in \cB$ of rank at least $rd$. 
\end{proof}

Using Remark~\ref{remark:coprime} and
Lemma~\ref{lem:cyclic_splitting}
we immediately obtain the following results.

\begin{lemma}[Regularity of blow-ups, non-constructive]
\label{lem:reg_non_cons}
For $\cB\leq M(n, \F)$, assume that $|\F| > (nd)^{\Omega(1)}$. Then 
$\rk(\rblowup{\cB}{d})$ is divisible by $d$.
\end{lemma}

\begin{lemma}[Regularity of blow-ups, constructive]
\label{lem_reg_blowup}
For $\cB\leq M(n, \F)$ and $\cA=\rblowup{\cB}{d}$, assume that 
$\mathrm{char}(\F)=0$ or 
$\mathrm{char}(\F)\nmid d$, and 
$|\F| > (nd)^{\Omega(1)}$. Then,
given a matrix $A\in 
\cA$ with $\rk A > (r-1)d$, there exists a deterministic algorithm that returns 
$\widetilde{A}\in \cA$ of rank $\geq rd$. 
This algorithm uses $\poly(nd)$ arithmetic 
operations and over $\Q$, all intermediate numbers have bit lengths polynomial in 
the input size.
\end{lemma}

To see all the ingredients together, we give an
expanded description in Algorithm~\ref{algo:round_up}.

\begin{algorithm}\label{algo:round_up}
  \caption{Algorithm \algo{RoundUpRank}}
  \KwIn{$\cB=\langle B_1, \dots, B_m\rangle\leq M(n, \F)$. $S\subseteq \F$ with 
  $|S|=(nd)^{\Omega(1)}$. $A\in \rblowup{\cB}{d}$ of $\rk(A)>(n-1)d$. If 
  $\fdchar(\F)=p>0$ then $p\nmid d$. }
  \KwOut{$A''\in \rblowup{\cB}{d}$ with $\rk(A'')=nd$.}

  $\F'\gets \F[\zeta]$, where $\zeta\gets$ a $d$th root of $\F$\;
  $X_1, Y\gets$ two independent formal variables\;
  $\Gamma=\{C_1, \dots, C_{d^2}\}\gets$ a $\F'(X_1, Y)$-basis of $M(d, \F'(X_1, 
  Y))$ such that $\F'(X_1, Y^d)$-linear span of $\Gamma$ is a central division 
  algebra over $\F'(X_1, Y^d)$\;
  Expand $A=\sum_{i\in[m], j\in[d^2]} \lambda_{i, j} B_i\otimes C_j$, 
  $\lambda_{i,j}\in \F'(X_1, Y)$ as a matrix in $\cB\otimes M(d, \F'(X_1, 
  Y))$\;
  Compute $A'=\sum_{i\in[m], j\in[d^2]} \mu_{i,j}B_i\otimes C_j$, $\mu_{i,j}\in S$ 
  such that $\rk(A')=nd\geq \rk(A)$\;
  $\Delta=\{E_{i,j}\mid i, j\in[d]\}\gets$ the standard basis of $M(d, \F'(X_1, 
  Y))$\;
  \tcp{That is $E_{i,j}(i,j)=1$, and other entries are $0$}
  Expand $A'=\sum_{i\in[m], j, k\in[d]}\nu_{i,j,k} B_i\otimes E_{j,k}$, 
  $\nu_{i,j,k}\in \F'(X_1, Y)$\;
  Compute $A''=\sum_{i\in[m], j, k\in[d]}\xi_{i,j,k} B_i\otimes E_{j,k}$, 
  $\xi_{i,j,k}\in S$ such that $\rk(A'')=\rk(A')$\;
  \Return $A''$.
\end{algorithm}

We remind 
the reader how our previous development supports 
Algorithm~\ref{algo:round_up}. For Line 1, Section~\ref{subsubsec:zeta} describes 
how to work with $\F[\zeta]$. 
Line 3 calls the procedure in 
Lemma~\ref{lem:cyclic_splitting} which in turn utilizes
Lemma~\ref{lem:cyclic_gen}. Line 4 and 7 are standard tasks in linear algebra. 
Line 5 and 8 rely on Lemma~\ref{lem:red_data}, which in turn calls 
Proposition~\ref{prop:rank}. Lemma~\ref{lem:correct_rank} ensures that 
the rank calculation is essentially correct.

We mention an immediate consequence.

\begin{corollary}\label{cor:regularity}
Let $\cB\leq M(n, \F)$. Assume that the characteristic of $\F$ is zero,
$d\geq n$, and we are given a matrix $A\in \rblowup{\cB}{d}$ 
with $\rk A = rd$. Then, for any $d'>d$ 
we can efficiently find $\widetilde{A}\in \rblowup{\cB}{d'}$ with $\rk 
\widetilde{A}\geq rd'$.
\end{corollary}
\begin{proof}
Induction on $d'$ and $d$. Assume first that $d'=d+1$. Then 
$\frac{d+1}{d}
\leq\frac{n+1}{n}<\frac{r}{r-1}$, whence $(d+1)(r-1)<dr$.
Therefore, if we embed $\rblowup{\cB}{d}$ into $\rblowup{\cB}{d+1}$
then $A$ has rank $rd>(r-1)(d+1)$ and Lemma~\ref{lem_reg_blowup}
 applies.
Proceed with $(d+2,d+1)$ in place of $(d+1,d)$.
\end{proof}

We close this subsection with two open problems.  
\begin{remark}
~
\begin{enumerate}
\item It is an interesting problem to investigate whether
(the non-constructive version of) the regularity lemma 
would hold over small fields. Giving a proof not relying
on division algebras might be a progress in this direction.
\item Corollary~\ref{cor:regularity} probably remains true for $d<n$. However, we 
do not see how to prove it.
\end{enumerate}

\end{remark}

\subsection{Incrementing rank via blow-up}

We have introduced in Section~\ref{subsec:wong} a key technique here, namely the 
second Wong sequences. Recall that given $A\in\cB\leq M(n, \F)$, the second Wong 
sequence can be used to detect whether there exists a $\cork(A)$-shrunk subspace. 
If such a shrunk subspace exists, then $\nrk(\cB)=\rk(A)$. The difficulty is the 
case 
when such $\cork(A)$-shrunk subspace does not exist. One natural idea to proceed 
is to find $A'\in\cB$ of rank $>\rk(A)$, and use $A'$ to test whether a 
$\cork(A')$-shrunk exists or not. If $\cB$ is rank-$1$ 
spanned, such an $A'$ does exist, and another ingredient in 
\cite{conf_version} is an 
update procedure that finds $A'$ of higher rank in this case. However, this in 
general is not possible, since 
$\rk(\cB)$ and $\nrk(\cB)$ may differ. Fortunately, the concept of blow-ups helps: 
instead of looking for $A'\in \cB$ of rank $> \rk(A)$, we shall look for $A'\in 
\rblowup{\cB}{d}$ of rank $> \rk(A)d$ for some not too large $d\in \N$. 
It turns out that this 
is achievable, and an application of the regularity lemma even yields $A''\in 
\rblowup{\cB}{d}$ of rank $\geq (\rk(A)+1)d$. 

\begin{theorem}
\label{thm-blup-incr}
Let $\cB\leq M(n, \F)$ and $\cA=\rblowup{\cB}{d}$. Assume that
we are given a matrix $A\in \cA$ with $\rk(A)=rd$. Let $d'$ be an integer $>r$. 
Suppose that $|\F|$ is $(ndd')^{\Omega(1)}$, and if $\fdchar(\F)=p>0$ then assume 
$p\nmid 
dd'$. There 
exists a deterministic algorithm that returns either an $(n-r)d$-shrunk 
subspace for $\cA$ (equivalently, an $(n-r)$-shrunk subspace for $\cB$), or a 
matrix $A^*\in \cA\otimes M(d', \F)$ of rank at least $(r+1)dd'$. This algorithm 
uses $\poly(ndd')$ arithmetic operations and, over $\Q$, all intermediate numbers 
have bit lengths polynomial in the input size.
\end{theorem}

\begin{proof}
We present the algorithm formally as Algorithm~\ref{algo:incre}. 

\begin{algorithm}
\caption{Algorithm \algo{IncrementRank}}\label{algo:incre}
  \KwIn{$\cB=\langle B_1, \dots, B_m\rangle\leq M(n, \F)$. $A\in 
  \cA=\rblowup{\cB}{d}$ of $\rk(A)=rd$. An 
    integer $d'>r$ such that if 
    $\fdchar(\F)=p>0$ then $p\nmid d'$. $S\subseteq \F$ with
        $|S|=(ndd')^{\Omega(1)}$. }
    \KwOut{One of the following: (1) An $(n-r)d$-shrunk subspace of $\cA$;
     (2)
    $A^*\in \rblowup{\cA}{d'}=\rblowup{\cB}{dd'}$ with $\rk(A^*)\geq (r+1)dd'$ . }
$W_0, W^*\gets \zvec$\;
$\ell=0$\;
\While{$W^*\subseteq\im(A)$ \textbf{and} $\ell\leq r+1$ }
{
$\ell\gets \ell+1$\;
$W_\ell\gets \cA A^{-1}(W_{\ell-1})$\;
$W^*=W_\ell$.
}
\If{$W^*\subseteq \im(A)$}
{\Return $A^{-1}(W^*)$.}
\tcp{In the following $\ell$ is the smallest integer $i$ such that 
$W_i\not\subseteq 
\im(A)$.}
Compute $C_1, \dots, C_\ell\in \cA$ such that $C_\ell A^{-1}C_{\ell-1}A^{-1}\dots 
C_1A^{-1}(0)\not\subseteq \im(A)$\;
Take $v_1\in A^{-1}(0)=\ker(A)$ such that $C_\ell A^{-1}C_{\ell-1}A^{-1}\dots 
C_1(v_1)\notin \im(A)$\;
\For{$i=2$ \textbf{to} $\ell$}{$v_i\gets A^{-1}C_{i-1}(v_{i-1})$.}
\tcp{Note that $C_\ell v_\ell\not\in\im(A)$.}
$\{u_1, \dots, u_{d'}\}\gets$ a basis of $\F^{d'}$\;
\For{$i=1$ \textbf{to} $\ell$}
{\If{$i=r+1=d'$}{
$Z_i\gets$ the matrix in $M(d', \F)$ such that
$Z_{r+1}(u_{r+1})=u_1$ and $Z_{r+1}(u_j)=0$ for $j\neq 
r+1$.}
\Else{
$Z_i\gets$ the matrix in $M(d', \F)$ such that 
$Z_i(u_i)=u_{i+1}$ and $Z_i(u_j)=0$ for $j\neq i$.}}

$A'\gets A\otimes I$\;
$C'\gets C_1\otimes Z_1+\dots +C_\ell\otimes Z_\ell$\;
\If{$\rk(C')>rdd'$}{$A''\gets C'$.}
\Else{Compute $\lambda \in S$ such that $A'+\lambda C'$ is of $>rdd'$\; $A''\gets 
A'+\lambda C'$.}
Compute $A^*$ of rank $\geq (r+1)dd'$ using $A''$ and Lemma~\ref{lem_reg_blowup}\;
\Return $A^*$.
\end{algorithm}

Let us first outline what the algorithm does. From Line 1 to 6 it 
computes the 
second Wong sequence with respect to $(A, \cA)$. Line 7 and 8 deal with the case 
when the 
sequence provides a $\cork(A)$-shrunk subspace. In the other case, we first 
utilize the sequence to get a matrix $A''$ of rank $>rdd'$ (Line 9 to 
25). Then we obtain the desired $A^*$, by applying 
the regularity lemma (Lemma~\ref{lem_reg_blowup}) to $A''$. 

We now explain some implementation details of the algorithm. 
\begin{description}
\item[Line 3, $\ell\leq r+1$] The second Wong sequence, when applied to matrix 
spaces of the form $\cA=\rblowup{\cB}{d}$, stabilizes faster because of the 
following. Since $(I\otimes 
M(d,\F))\cA=\cA$, at stage $j$ we have $(I\otimes M(d,\F))\cA W_j=\cA W_j$, whence 
the dimension
of $\cA W_j$ is divisible by $d$ for every $j$. It follows that, until 
stabilization, the dimension of $\cA W_j$ increases by at least $d$ and so the 
sequence stabilizes to its limit in at most $r+1$ steps when applied to $A\in 
\cA$ of rank $rd$.
\item[Line 9] To compute $C_i$'s, one perform the following. Take a basis of 
$\cA$. Search for a basis element $Y$ such that $YA^{-1}(\cA 
A^{-1})^{\ell-1}(0)\not\subseteq \im(A)$. Put $C_\ell=Y$ and search for a basis 
 element $Y$ such that $C_\ell A^{-1} Y A^{-1}(\cA 
 A^{-1})^{\ell-2}(0)\not\subseteq 
\im(A)$. Continue this iteration and the desired $C_i$'s can be computed. 
\end{description}

That Algorithm~\ref{algo:incre} runs in the stated time bound follows easily from 
Fact~\ref{fact:Wong} and Lemma~\ref{lem_reg_blowup}. To see the correctness of the 
algorithm, it remains to prove that in Line 21 to 25 we do obtain $A''$ of 
rank $>rdd'$. Consider the vectors $w_1=v_1\otimes u_1$, 
$\ldots$, $w_t=v_t\otimes u_t$. We now observe that: (1) $w_1\in \ker A'$; (2) 
$A'w_j=C'w_{j-1}$ for
$j=2,\ldots,t$; (3) $C'w_t=C_tv_t\otimes 
u_{t+1}\not\in (\cA\F^{nd})\otimes \F^{d'}$; as $\cA\F^{nd}\otimes \F^{d'}\supseteq
A'\F^{ndd'}$, we have $C'w_t\not\in A'\F^{ndd'}$. This means that the limit of the 
second Wong sequence for the pair $(A', \langle A', C'\rangle)$ runs out of the 
image of
$A'$. By Fact~\ref{fact:dimtwo}, $A'$ is not of maximum rank in 
$\langle A', C'\rangle$, and Line 21 to 25 just describe a straightforward method 
to obtain a matrix of highest rank in a $2$-dimensional matrix space. 
\end{proof}

An iteration based on Theorem~\ref{thm-blup-incr} proves the following 
Theorem~\ref{thm-blup-main}. 
Note that in Theorem~\ref{thm-blup-incr}, $d'$ can be 
chosen as either $r+1$ or $r+2$, depending on which is not divisible by 
$\mathrm{char}(\F)$.

\begin{theorem}
\label{thm-blup-main}
Suppose we are given $\cB:=\langle B_1, \dots, B_m\rangle\leq M(n, \F)$, and 
$A\in\cB$ with $\rk(A)=s<n$. Let $d=(n+1)!/(s+1)!$, and assume that 
$|\F|=\Omega(nd)$. Then there exists a deterministic algorithm, that 
computes a matrix $B\in \cB\otimes M(d', \F)$ of rank $rd'$ for some $d'\leq d$ 
and, if $r<n$, an $(n-r)$-shrunk subspace for $\cB$. The 
algorithm uses $\poly(n, d)$ arithmetic operations, and when working over $\Q$, 
has bit complexity polynomial in $n$, $d$ and the input size.
\end{theorem}

Now Corollary~\ref{cor:blup_edmonds} and~\ref{cor:blup_degree} follow easily. To 
see Corollary~\ref{cor:blup_edmonds}, note that if
we choose a matrix in $\cB$ randomly, it will be of maximum rank. Using that 
matrix as $A$ in Theorem~\ref{thm-blup-main}, Corollary~\ref{cor:blup_edmonds} is 
proved. For Corollary~\ref{cor:blup_degree}, if $\cB$ has no shrunk subspace, 
a full-rank matrix will be certainly present in $M(d', \F)\otimes \cB$ for some 
$d' \leq (n+1)!$, giving us the upper bound on $\sigma(R(n, m))$.

\paragraph{Acknowledgements.} We would like to thank M{\'a}ty{\'a}s Domokos, Bharat 
Adsul, Ketan Mulmuley, Partha Mukhopadhyay and K. N. Raghavan for discussions 
related to this work. 
Part of the work was done when Youming was visiting the Simons Institute for the 
program Algorithms and Complexity in Algebraic Geometry and when G\'abor
was visiting the 
Centre for Quantum Technologies,
 National University of Singapore. 
Research of the first 
author was also supported in part by the 
Hungarian National Research, Development and Innovation Office -- NKFIH,
Grants NK105645 and 115288. 
Youming's 
 research was supported by the Australian Research Council DECRA DE150100720. 
 
\bibliographystyle{amsplain}

\bibliography{references}

\appendix

\section{Derksen's bound applied to $R(n, m)$}\label{app:derksen}

\begin{proof}
We just need to indicate certain parameters for the matrix semi-invariants that 
are used in Derksen's bound. 

Suppose a group $G$ acts on a vector space $V$ rationally, and let $R$ be the 
resulting invariant ring. Theorem 1.1 in \cite{derksen_bound} shows that the 
degree bound is upper bounded by $\max(2, 3/8\cdot s\cdot \sigma(R)^2)$, where 
$\sigma$ is the degree bound for defining the nullcone, and $s$ is the dimension 
of $R$. 

$s$ is upper bounded by the number of variables. Therefore for $R(m, n)$, $s\leq 
mn^2\leq n^4$. 

To bound $\sigma$, we use Proposition 1.2 in \cite{derksen_bound}. Recall that $G$ 
as an algebraic group, is defined by a system of polynomial 
equations in $z_1, \dots, z_t$. For example, $\SL(n, \F)\times \SL(n, \F)$ is 
defined by $\det(X)=1$ and $\det(Y)=1$, where $X$ and $Y$ are $n\times n$ variable 
matrices. The action of $G$ is rational, so it can be recorded as 
$\rho:G\to\GL(V)$ by $g\to (a_{i,j}(g))_{i,j\in\dim(V)}$, where each $a_{i,j}$ is 
a polynomial in $z_1, \dots, z_t$.

$\sigma(R)$ is then upper bounded by $H^{t-m}A^m$, where $t$ is the number of 
variables used to define $G$ as above, $m=\dim(G)$, $H$ is the maximum degree over 
polynomials defining $G$, and $A$ is the maximum degree over polynomials defining 
the action. So for $R(n, m)$, $t=2n^2$, $m=2n^2-2$, $H=n$, and $A=2$. It follows 
that $\sigma(R(n, m))\leq n^2\cdot 2^{2n^2-2}$. 

Therefore $R(n, m)$ is generated by elements of degree $\leq 3/128\cdot n^8\cdot 
16^{n^2}$.
\end{proof}

\end{document}